\documentclass{article}
\usepackage{amsfonts,color}
\usepackage{amsmath,mathabx}
\usepackage{caption}
\usepackage{comment}
\usepackage{enumitem}
\usepackage{lscape}
\usepackage{subcaption}
\usepackage{graphicx}
\usepackage{hyperref}
\usepackage{cleveref}
\usepackage{setspace} 
\usepackage{tikz}
\usepackage{tabularx}
\usepackage{adjustbox}
\usepackage{marginnote}
\usepackage{tablefootnote} 
\usepackage{url}
\usetikzlibrary{patterns}
\newtheorem{theorem}{Theorem}

\newtheorem{corollary}[theorem]{Corollary}

\newtheorem{definition}[theorem]{Definition}
\newtheorem{example}[theorem]{Example}

\newtheorem{lemma}[theorem]{Lemma}

\newtheorem{proposition}[theorem]{Proposition}

\newtheorem{remark}[theorem]{Remark}

\newenvironment{proof}[1][Proof]{\noindent\textbf{#1.} }{\ \rule{0.5em}{0.5em}}

\newcommand{\R}{\mathbb{R}}
\newcommand{\E}{\operatorname{\mathbb{E}}}

\newcommand{\F}{\mathcal{F}}
\newcommand{\M}{\mathcal{M}}
\newcommand{\N}{\mathcal{N}}

\newcommand{\trho}{\tilde{\rho}}

\DeclareMathOperator*{\esssup}{ess\,sup}

\newcommand{\keywords}[1]{\textbf{Keywords } #1}

\newcommand{\sqb}[1]{\ensuremath{ \left[ #1 \right] }}
\newcommand{\of}[1]{\ensuremath{\left( #1 \right)}}

\title{\vskip -1.4cm Generalized Orlicz premia\thanks{This paper was circulated earlier under the title ``Geometrically convex return risk measures and Orlicz premia''.
This research was funded in part by the Netherlands Organization for Scientific Research under an NWO-Vici grant 2020--2027 (Ayg\"un and Laeven).}\\[0mm] 
 }
\author{M\"ucahit Ayg\"un \\
{\footnotesize Dept.~of Quantitative Economics}\\
{\footnotesize University of Amsterdam}\\
{\footnotesize and Tinbergen Institute}\\
{\footnotesize \texttt{M.Aygun@uva.nl}}\\
\and Fabio Bellini \\
{\footnotesize Dept.~of Statistics and Quantitative Methods}\\
{\footnotesize University of Milano-Bicocca}\\
{\footnotesize \texttt{fabio.bellini@unimib.it}}\\
\and Roger J.~A.~Laeven\footnote{Corresponding author.}\\
{\footnotesize Dept.~of Quantitative Economics}\\
{\footnotesize University of Amsterdam, CentER}\\
{\footnotesize and EURANDOM}\\
{\footnotesize \texttt{R.J.A.Laeven@uva.nl}}
}

\date{\today}
\begin{document}
	
\maketitle

\begin{abstract}
We introduce a generalized class of Orlicz premia based on possibly nonconvex 
loss functions, extending the classical framework of Haezendonck and Goovaerts 
(1982).\ Without the usual convexity requirement on the loss function $\Phi$, the Orlicz framework 
naturally encompasses quantiles, expectiles and $L^p$-quantiles while preserving 
the fundamental properties of Orlicz premia. 
We show that within this framework cash-additivity axiomatizes 
$L^p$-quantiles, generalizing the classical `collapse-to-the-mean' result for cash-additive convex Orlicz premia into a `collapse-to-$L^p$-quantiles' result, with expectiles as a special case. 
We focus on two natural classes of nonconvex loss functions: 
\emph{concave-convex} Orlicz functions, which mimic the idea of 
S-shaped value functions in prospect theory, and \emph{GA-convex} Orlicz functions, which can be described in terms of comparative convexity with respect to a logarithmic reference, and for which the 
corresponding Orlicz premium is geometrically convex. 
Finally, we show that a suitable subclass of generalized Orlicz premia coincides with the class of law-invariant, monotone, positive, positively homogeneous, normalized functionals that are weakly lower semicontinuous, continuous from above, and whose level 
sets are convex with respect to mixtures (the so-called CxLS property). 
\end{abstract}
\medskip 

.

\section{Introduction}
A rich literature in actuarial science and financial mathematics has studied the theory of measures of risk, both in static and dynamic environments.
Risk measures are used for several different purposes including the calculation of insurance premia. 
Among premium calculation principles, insurance premia based on Orlicz (Luxembourg) norms---Orlicz premia---play a prominent role.\smallskip

Orlicz premia were introduced in the actuarial literature by Haezendonck and Goovaerts in \cite{HG82}, the very first issue of \textit{Insurance:\ Mathematics and Economics}.\ 
They constitute a class of extensively studied premium principles, defined by
$$
H_\Phi(X) := \inf \{k > 0 | \E \left [\Phi(X/k ) \right] \leq 1 \},
$$
where $X$ is a random variable representing the loss and the so-called \emph{Young} function $\Phi \colon [0, +\infty) \to [0, +\infty)$
is increasing, convex and satisfies $\Phi(0)=0$ and $\Phi(1)=1$. 
Orlicz premia represent multiplicative versions of the \emph{zero utility} premium, since it often (i.e., under regularity conditions) holds that 
$$
\E \left [\Phi \left(\frac{X}{H_\Phi(X)} \right) \right] =1,
$$
whereas, by comparison, the indifference equation defining the zero utility premium $\Pi(X)$ for a utility function of wealth $u$ is given by
$$
\E \left[u \left(X - \Pi(X) \right) \right]=0. 
$$

Orlicz premia are positively homogeneous, i.e., $H_\Phi(\lambda X)=\lambda H_\Phi(X)$, for $\lambda\geq 0$, and satisfy many appealing mathematical properties (\cite{HG82,GDH84,RR91}), since they are a generalization of $L^p$-norms that arise in the special case $\Phi(x)=x^p$, i.e., power/CRRA (dis)utility, with a rich duality theory that has been thoroughly studied also in the financial mathematics and optimization literature (see, e.g., \cite{BF08,CL08,CL09,LS14}). 
\smallskip

Orlicz premia are also the basis for the definition of Haezendonck-Goovaerts risk measures, introduced in \cite{GKDT04} by a construction resembling the optimized certainty equivalent of Ben-Tal and Teboulle (\cite{BT86, BT07}), and successively studied in \cite{BRG08, BRG12} and extended to a robust and dynamic  setting in \cite{BLR18, BLR21}.  
\smallskip

Our starting point is the observation that relaxing the convexity requirement on the loss function $\Phi$ yields a generalized notion of Orlicz premium that includes well-known possibly nonconvex functionals such as quantiles, expectiles, $L^p$-quantiles (\cite{C96}), and other forms of generalized quantiles as in \cite{BKMR14}. 
\smallskip

Besides encompassing these relevant functionals, the idea of relaxing the convexity assumption on $\Phi$ is motivated by classical papers in utility theory such as \cite{FS48} and \cite{KT79}, advocating the use of concave/convex utility functions of wealth or S-shaped value functions of gains/losses with respect to a possibly exogenous reference point, to model the simultaneous presence of risk aversion and risk-loving behavior.  
Furthermore, we mention that the weakening of the convexity assumption on $\Phi$ has also been pursued in the recent mathematics literature, in the context of spaces and topologies, as we discuss in Section~\ref{sec:Orl}. 
Another related development is the analysis of state-dependent Orlicz premia as in the recent \cite{S25}. 
\smallskip

Relaxing the definition of the Young function has a few consequences, which we discuss in Section~\ref{sec:Orl}. 
Most importantly, we show in Proposition~\ref{prop:orl} that the basic properties of Orlicz premia still hold in the general case, with the exception of convexity that holds if and only if $\Phi$ is convex as in the classical case. 
\smallskip

Our first main result is Theorem~\ref{th:expectiles}, showing that a cash-additive generalized Orlicz premium is necessarily an $L^{p+1}$-quantile, in particular with $p=1$ (i.e., corresponding to an expectile) if additionally convexity or concavity is assumed. 
This result gives a novel point of view on $L^p$-quantiles and expectiles, and further motivates their actuarial applications. 
Further, Theorem~\ref{th:expectiles} extends the classical result about `collapse to the mean' of cash-additive, convex Orlicz premia that was obtained in \cite{HG82,GDH84} under more restrictive assumptions on the Young function $\Phi$ to a more general `collapse to $L^p$-quantiles' or expectiles.
\smallskip

We then focus on two specific general families of loss functions. 
The first is constituted by \emph{concave-convex} loss functions, which naturally extend the cases of $L^p$-quantiles and $L^{p,q}$-quantiles, that have the form 

$$
\Phi(x)= 1+ \alpha \Phi_1 \left ((x-1)_+ \right) - (1-\alpha) \Phi_2 \left ((x-1)_- \right),
$$
where $\Phi_1$ and $\Phi_2$ are normalized convex Young functions. 
For this class, we derive a few properties of the corresponding generalized Orlicz premia, and we study their consistency with respect to stochastic orders, extending the result on the consistency of $L^p$-quantiles with respect to $p$-convex orders found in \cite{B12}. \smallskip

The second class is constituted by GA-convex Orlicz functions, a notion arising in the algebraic theory of convexity based on the comparison of the Geometric and the Arithmetic mean. 
Our interpretation is in terms of comparative convexity, since these Orlicz functions satisfy the inequality

$$
\frac{\Phi''(x)}{\Phi'(x)} \geq -\frac{1}{x},
$$
and are thus more convex than the logarithmic (dis)utility function. 
Recall that comparative convexity can arise also between nonconvex functions, and that convex functions are by definition more convex than linear functions. 
In Proposition~\ref{prop:gg-GA}, we show that in this case the corresponding Orlicz premia satisfy another form of algebraic convexity called \emph{geometric convexity} or GG-convexity for short. 
\smallskip

Finally, in Section~\ref{sec:CxLS} we present the central result of the paper, that is, Theorem~\ref{th:axiom}, in which we give an axiomatic foundation to generalized Orlicz premia, showing that a suitable subclass of them corresponds precisely to those functionals satisfying law-invariance, monotonicity, positive homogeneity, positivity, normalization, weak lower semicontinuity, continuity from above and convex level sets with respect to mixtures, the so-called \emph{CxLS} property. 
This is relevant because the CxLS property is a necessary condition for elicitability, as has been extensively discussed in the literature (see, e.g.,\ \cite{O85,G11,DBBZ14,Z16,FZ16}).

The paper is organized as follows.\ In Section~\ref{sec:Orl} we introduce generalized Orlicz premia and discuss their properties. 
In Section~\ref{sec:gg-conv} we focus on concave-convex and GA-convex Orlicz functions.  
In Section~\ref{sec:CxLS} we provide the central axiomatization result. Section~\ref{sec:con} concludes.  
All proofs are in the Appendix.

\section{A generalization of Orlicz premia}\label{sec:Orl}

In order to introduce the definition of generalized Orlicz premia, we start by recalling from the mathematics literature the notions of Young function, Luxemburg norm and Orlicz space, following the classical references \cite{RR91} and \cite{ES92}. \\  
We consider random variables  
defined on a common nonatomic probability space $(\Omega, \F, P)$.\ Equalities and inequalities between random variables are meant to hold $P$-a.s.\ without further mentioning.
\begin{definition}
\label{def:orl-old}
Let $\Phi \colon [0,+\infty) \to [0, +\infty]$ satisfy the following properties:
\begin{enumerate}[label=\alph*)]
\item  $\Phi(0)=0$, $\Phi(1)=1$ 
\item  $\Phi(+\infty)=+\infty$
\item $\Phi$ is nondecreasing
\item  $\Phi$ is left-continuous
\item  $\Phi$ is convex. 
\end{enumerate}
Then $\Phi$ is called a \emph{normalized Young function} 
and 
\begin{equation}
\label{eq:orl-space}
    L^{\Phi}:= \{ X \in L^0(\Omega, \F, P) \mid \E[\Phi(\vert X \vert /k)] \leq 1 \text{ for some } k>0 \} 
\end{equation}
is the associated Orlicz space. For 
$X \in L^\Phi$, the \emph{Orlicz premium} is 
\begin{equation}
\label{eq:orl-prem}
H_\Phi(X):= \inf \{k >0 \mid \E[\Phi(\vert X \vert /k)] \leq 1 \}. 
\end{equation}
\end{definition}

There are actually minor differences between the definitions of Young function across various references in the mathematics literature (e.g., \cite{RR91} and \cite{ES92}) and the actuarial literature (e.g., \cite{HG82}). Definition \ref{def:orl-old} above subsumes elements of both and covers all relevant cases. We make a few comments useful for comparison with the more general version that we are going to introduce. 
\smallskip

First, allowing $\Phi$ to take the value $+\infty$ is customary in the mathematics literature, and is motivated by the inclusion of the case $H_\Phi(X)=\esssup(X)$. 
Second,
the normalization $\Phi(1)=1$ is often omitted in the mathematics literature, but is sensible from the actuarial point of view since it guarantees the constancy property $H_\Phi(c)=c$, for every $c \in [0, +\infty)$.\  
Third, condition e) implies that $\Phi$ is continuous wherever it is finite, so only if $\Phi(c)=+\infty$ for some $c \in [0, +\infty)$ condition d) is not trivially satisfied. 
Finally, in the actuarial literature, premium principles are computed for nonnegative variables representing losses, so in equations \eqref{eq:orl-space} and \eqref{eq:orl-prem} the absolute values may be removed.
\smallskip

We begin by illustrating that the notion of a Young function can be significantly extended, obtaining as corresponding Orlicz premia
functionals that are well-known in actuarial applications and in the theory of risk measures.

\begin{example}[Quantiles] \label{ex:quantiles}
Let $X \geq 0$ and let
\[
\Phi_\alpha (x) := 
\begin{cases}
\alpha &\text { if } 0 \leq x \leq 1,\\
1+\alpha &\text { if } x > 1 ,
\end{cases}
\]
with $0 < \alpha \leq 1$. Then, 
\begin{align*}
H_{\Phi_\alpha}(X) &= \inf \left \{ k >0 \;  | \; \E \left [  \Phi_\alpha \left (X/k \right)  \right ] \leq 1 \right \} \\
&= \inf \left \{ k >0 \;  | \; \alpha P (X \leq k) + (1+\alpha) (1- P(X \leq k)) \leq 1 \right \} \\
&=  \inf \left \{ k >0 \;  | \; P(X \leq k) \geq \alpha \right \}\\
&=q_\alpha^-(X),
\end{align*}
which is the left $\alpha$-quantile of $X$.\
Notice that if $\alpha=1$, then $H_{\Phi_\alpha}(X)=\esssup(X)$.\ 
\end{example}

\begin{example}[Expectiles] \label{ex:expectiles}
Let $X \geq 0$ with $ \E[X] < +\infty$ and let 
$$\Phi_\alpha(x)=1+ \alpha(x-1)_{+} -(1-\alpha)(x-1)_{-},$$ 
with $0 < \alpha < 1$. 
Then,
\begin{align*}
H_{\Phi_\alpha}(X) &= \inf \left \{ k >0 \;  | \; \E \left [  \Phi_\alpha \left (X/k \right)  \right ] \leq 1 \right \} \\
&= \inf \left \{ k >0 \;  | \; 1+ \alpha \E [ (X/k-1)_+ ]-(1-\alpha) \E[(X/k-1)_- ] \leq 1 \right \} \\
&=  \inf \left \{ k >0 \;  | \; \alpha\E \left [ (X-k)_+  \right] \leq (1-\alpha) \E[(X-k)_-] \right\} \\
&=e_\alpha(X),
\end{align*}
which is the $\alpha$-expectile of $X$ as introduced in \cite{NP87}. 


\end{example}
\begin{example}[$L^{p+1}$-quantiles]\label{ex:lp-quantiles}
Let $0<\alpha<1$, $p>0$ and $$\Phi_{\alpha,p}(x)=1+ \alpha(x-1)_{+}^{p} -(1-\alpha)(x-1)_{-}^p.$$
Let $X \geq 0$ with $\E[X^p] < +\infty.$
Then, as before, 
\begin{align*}
H_{\Phi_{\alpha,p}}(X) &= \inf \left \{ k >0 \;  | \; \E \left [  \Phi_{\alpha,p} \left (X/k \right)  \right ] \leq 1 \right \} \\
&=  \inf \left \{ k >0 \;  | \; \alpha\E \left [ (X-k)_+^p  \right] \leq (1-\alpha) \E[(X-k)_-^p] \right\}\\
&=z_{\alpha,p}(X),
\end{align*}
which is an $L^{p+1}$-quantile in the terminology of \cite{C96}, generalizing the case of expectiles that arises if $p=1$. 
\end{example}


In Example \ref{ex:quantiles}, the function $\Phi$ does not satisfy a), b) and e); in Example \ref{ex:expectiles}, it does not satisfy a) and satisfies e) if and only if $\alpha \geq 1/2$; and in Example \ref{ex:lp-quantiles},  a) and e) are also in general not satisfied. 
This leads us to the notion of \emph{generalized Orlicz function} introduced in the following definition. 

\begin{definition}\label{def:orl}
Let $\Phi \colon [0,  +\infty) \to [0,+\infty)$ satisfy:
\begin{enumerate}[left=0pt, itemsep=0pt]
\item [a)] $\Phi(x) \leq 1$ if $x\leq 1$, $\Phi(x)>1$ if $x>1$
\item [b)] $\Phi$ is nondecreasing 
\item [c)] $\Phi$ is left-continuous. 
\end{enumerate}
Then we say that $\Phi$ is a \emph{generalized Orlicz function}, and we define the corresponding generalized Orlicz space $L^\Phi$  and the corresponding Orlicz premium $H_\Phi$ as in \eqref{eq:orl-space} and \eqref{eq:orl-prem}, respectively. 
\end{definition} 

It is immediate to check that Definition \ref{def:orl} is an extension of Definition~\ref{def:orl-old} and covers Examples~\ref{ex:quantiles}, \ref{ex:expectiles} and~\ref{ex:lp-quantiles} as special cases.\ Three features are worth commenting on. 
First, relaxing $\Phi(0)=0$ to $\Phi(0)\leq 1$ makes it unnecessary to allow $\Phi=+\infty$: the essential supremum, which in Definition~\ref{def:orl-old} required this convention, is now obtained as the case $\alpha=1$ of Example~\ref{ex:quantiles}, which has $\Phi(0)=1$.\ Accordingly, $\Phi$ is assumed to be finite throughout Definition~\ref{def:orl}.\ Second, condition a) weakens the classical normalization $\Phi(0)=0$ and $\Phi(1)=1$, but is still strong enough to ensure the constancy property of $H_\Phi$, as we will see in Proposition~\ref{prop:orl} below. 
Finally, a novel phenomenon that occurs with the generalized definition is that there can be different Orlicz functions that yield the same Orlicz premium, although they must differ by an affine transformation as the following lemma shows. 
\smallskip

\begin{lemma}\label{lem:equiv-orl}
Let $\Phi_1,\Phi_2$ be generalized Orlicz functions in the sense of Definition~\ref{def:orl}. Then $H_{\Phi_1}(X)=H_{\Phi_2}(X)$ for every $X\in L^\infty_+$ if and only if there exists $c>0$ such that
\begin{equation}\label{eq:affine-equiv}
\Phi_2(x)=1+c\,\bigl(\Phi_1(x)-1\bigr),\text { for all } x\in[0,+\infty).
\end{equation}
\end{lemma}

\begin{definition}
If \eqref{eq:affine-equiv} holds, we say that $\Phi_1$ and $\Phi_2$ are \emph{equivalent Orlicz functions} and we write 
$
\Phi_1 \sim \Phi_2. 
$
\end{definition}
This phenomenon does not occur for classical Young functions, since the additional requirement that $\Phi(0)=0$ makes the representing Young function unique.\smallskip 

The convexity of the Young function $\Phi$ implies that $H_{\Phi}$ is a norm and the space $L^{\Phi}$ a Banach space.
In the mathematics literature, there have been several directions in which Definition~\ref{def:orl-old} has been generalized. 
First of all, the early \cite{M83} considers the general case of a state-dependent Young function 
$$
\Phi \colon \Omega \times [0,+\infty) \to [0,+\infty]. 
$$
Second, the recent monograph \cite{HH19} argues that it might be sensible to weaken the requirement of convexity to the requirement that $\Phi$ is almost increasing of order $1$, which in their terminology means that there exists $c \geq 1$ such that
$$
\Phi(x)/x \leq c \cdot \Phi(y)/y, \text { for all } x<y.
$$

Finally, the classical monograph \cite{RR91} considers in Chapter 10 the case of a general nondecreasing left-continuous $\Phi$ with $\Phi(0)=0$, but with a different definition of the norm as follows:
$$
\Vert X \Vert_\Phi= \inf \{k >0 \mid \E[\Phi(\vert X \vert /k)] \leq k \},
$$
where the choice of $\le k$ in place of $\le 1$ is motivated by the need to maintain the triangle inequality without convexity.\smallskip

From a mathematical point of view, the simplest example of a generalized Orlicz function in the sense of Definition \ref{def:orl}
is $\Phi(x)= x^p$ with $0<p<1$, where the Orlicz space $L^\Phi$ becomes an $L^p$ space with $0<p<1$ and 
$
H_\Phi(X)=\E[\vert X \vert^p]^{1/p}
$.  
As is well-known, in this case $H_\Phi$ is not a norm since it fails the subadditivity property.\ However, such $L^p$ spaces are complete metric vector spaces with the metric 
$
d(X,Y):=H_\Phi(\vert X - Y \vert).  
$
Further, these spaces are not locally convex, and as a consequence they do not have nontrivial linear continuous functionals. \smallskip

More generally, $L^\Phi$ is still a vector space under fairly general conditions. 

\begin{lemma}\label{lemma:vec}
Let $\Phi$ be a generalized Orlicz function in the sense of 
Definition~\ref{def:orl} and assume that $\Phi(0)=0$, and that $\Phi$ is concave and right-continuous at $0$. Then $L^\Phi$ is a vector space.
\end{lemma}






The properties required to the Orlicz function $\Phi$ in Definition~\ref{def:orl} are necessary to preserve the fundamental properties of Orlicz premia, with the exception of convexity.
In order to state them we recall the following standard terminology, adapted to the present case of functionals defined on positive random variables. 
\begin{definition}
A functional $\rho \colon \mathcal{X} \to \R$ on a domain $L^\infty_+(\Omega, \F, P) \subseteq \mathcal{X} \subseteq L^0_+(\Omega, \F, P)$ is 
\begin{itemize}
\item monotone, if $X \leq Y \Rightarrow \rho(X) \leq \rho(Y)$
\item positive, if $X \geq 0$ and $P(X>0)>0$ implies that $\rho(X) > 0$
\item positively homogeneous, if $\rho( \lambda X)= \lambda \rho(X), \, \text { for all } \lambda \geq 0$ 
\item cash-additive, if $\rho(X+c) = \rho(X) + c$, for all $c \in [0,+\infty)$

\item cash-sub(super)additive, if $\rho(X+c) \leq  (\geq)\, \rho(X) + c$, for all $c \in [0,+\infty)$


\item law-invariant, if $X \overset{d}{=} Y$ implies that $\rho(X)=\rho(Y)$, where $\overset{d}{=}$ stands for equality in distribution of random variables. 
\end{itemize}
\end{definition}
We denote by $\M([0,+\infty))$ the set of positively supported probability measures and by $\M_c([0,+\infty))$ its subset of compactly supported measures. 
Law-invariant functionals induce functionals on $\M([0,+\infty))$ by setting $\rho(F):= \rho(X)$, with $X \sim F$. 
Nonatomicity of the underlying probability space $(\Omega, \F, P)$ ensures the existence of a random variable with any distribution $F$.

\begin{proposition}\label{prop:orl}
Let $\Phi$ and $H_\Phi(X)$ be as in Definition~\ref{def:orl}. 
Then:
\begin{enumerate}[label=\alph*), left=0pt, itemsep=0pt]
\item $H_\Phi$ is monotone, positively homogeneous and satisfies $H_\Phi (c) = c$, for all $c \in [0,+\infty)$
\item If $H_{\Phi}(X)>0$, then $H_\Phi(X)= \min \{k>0 \mid \E[\Phi(X/k)] \leq 1 \}$
\item  It holds that $H_\Phi(X) \leq 1 \iff \E[\Phi(X)] \leq 1$
\item  If $\Phi$ is strictly increasing and continuous and $X \in L^\infty(\Omega, \F, P)$, then $H_\Phi$ is the unique solution of the equation
$
\E\left [ \Phi \left (X/H_{\Phi} \right) \right] = 1
$
\item If $\Phi$ is convex, then $H_\Phi$ is convex 
\item If there exists $u < 1$ with $\Phi (u) < 1$, then $H_\Phi$ is convex only if $\Phi$ is convex
\item $H_\Phi$ is law-invariant in the sense that
$
X \overset{d}{=} Y \Rightarrow H_\Phi(X) = H_\Phi(Y).
$
\end{enumerate}
\end{proposition}

\smallskip

We note the following.
If the Orlicz function $\Phi$ does not 
satisfy left-continuity, then properties c) and d) in Proposition~\ref{prop:orl} may not hold. 
Property d) can also fail if $X$ is not essentially bounded. 
Notice also that $H_\Phi(X)=1$ does not necessarily imply that $\E[\Phi(X)]=1$.\ The requirement that there exists $u<1$ with $\Phi(u)<1$ to prove the only if condition for convexity in item f) is automatically satisfied by Young functions, but has to be added in the generalized Orlicz case. 
Indeed, as item a) of the following lemma shows, any possibly nonconvex generalized Orlicz function with $\Phi(0)=1$ yields the essential supremum, which is convex. 
\smallskip

The following lemma characterizes positive generalized Orlicz premia. 
\begin{lemma}\label{lem:positivity}
Let $\Phi$ and $H_\Phi$ be as in
Definition~\ref{def:orl}. 
Then exactly one of the following cases holds:
\begin{enumerate}[left=0pt, itemsep=0pt]
\item[a)] if $\Phi(0)=1$, then $H_\Phi(X)=\esssup X$, and $H_\Phi$ is positive
\item[b)] if $\Phi(0)<1$ and $\Phi(+\infty)=+\infty$, 
then $H_\Phi$ is positive
\item[c)] if $\Phi(0)<1$ and $\Phi(+\infty)<+\infty$, 
then $H_\Phi$ is not positive.
\end{enumerate}
\end{lemma}

\smallskip

Notice that the usual quantiles considered in Example \ref{ex:quantiles} are not positive, with the exception of the case $\alpha=1$ that corresponds to the essential supremum. 

\smallskip

A classical result about convex Orlicz premia is that they are cash-additive if and only if they collapse to the mean, as has been proved in \cite{HG82, GDH84} under the additional assumption that the Young function $\Phi$ is differentiable.\ 
Remarkably, enlarging the class of loss functions as in Definition~\ref{def:orl} enlarges the class of cash-additive Orlicz premia as established in the following.

\begin{theorem}\label{th:expectiles}
Let $\Phi$ and $H_\Phi(X)$ be as in Definition~\ref{def:orl} with $\Phi$ continuous and strictly increasing.\ 
Then:
\begin{itemize}
\item [a)] $H_{\Phi}$ is cash-additive if and only if there exist $\alpha \in (0,1)$ and $p > 0$ such that
\begin{equation*}
\Phi(x) \sim 1+\alpha(x-1)_+^p-(1-\alpha)(x-1)_-^p,
\end{equation*}
which gives $H_{\Phi}(X)=z_{\alpha,p}(X)$. 
\item [b)] $H_{\Phi}$ is cash-additive and convex (resp.\ concave) if and only if
\begin{equation*}
\Phi(x) \sim 1+\alpha(x-1)_+ -(1-\alpha)(x-1)_-,
\end{equation*}
which gives $H_{\Phi}(X)=e_\alpha(X)$, with $\alpha \geq 1/2$ (resp.\ $\alpha \leq 1/2$). 
\end{itemize}
\end{theorem}

We conclude that a cash-additive generalized Orlicz premium is an $L^p$-quantile that is necessarily an expectile if additionally convexity or concavity holds. 

\begin{remark}
The usual quantile considered in Example~\ref{ex:quantiles} is excluded by the strict
monotonicity and continuity assumptions in Theorem~\ref{th:expectiles}, although it would
correspond formally to the limiting case $p=0$ in the family
$\Phi(x)=1+\alpha(x-1)_+^p-(1-\alpha)(x-1)_-^p$, and would satisfy cash-additivity. 
Our proof technique cannot be applied without the assumptions of strict monotonicity and continuity of $\Phi$. 
It might be possible to derive a more general formulation of Theorem~\ref{th:expectiles} including the usual quantiles, but we do not pursue this objective here. 
\end{remark}

We end the section with an interesting example in which cash-additivity fails, but cash-subadditivity or cash-superadditivity is satisfied.  

\begin{example}[$L^{p,q}$-quantiles]
\label{ex:lpq-quantiles}
An interesting generalization of Example \ref{ex:lp-quantiles} is given by the family of Orlicz functions
$$
\Phi(x) = 1 + \alpha(x-1)_+^p - (1-\alpha)(x-1)_-^q, 
$$
with $0<\alpha<1$ and $p,q > 0$.\ 
As a consequence of Theorem~\ref{th:expectiles}, cash-additivity holds if and only if $p=q$, as can be verified directly. 
Note that the $L^{p,q}$-quantile is cash-subadditive if $p\geq q$ and cash-superadditive if $p\leq q$. 
Indeed, letting $H_{\Phi}(X)=k$, then d) in Proposition~\ref{prop:orl} implies that
\[ \alpha \E\sqb{\of{\frac{X-k}{k}}_+^p}=(1-\alpha)\E\sqb{\of{\frac{X-k}{k}}_-^q}.\]
For any $c\geq 0$, assuming $p\geq q$ we have $\of{\frac{c+k}{k}}^{p-q}\geq 1$, which gives
\begin{align*}
    & \alpha \E\sqb{\of{\frac{X-k}{k}}_+^p}\leq (1-\alpha) \E\sqb{\of{\frac{X-k}{k}}_-^q} \of{\frac{c+k}{k}}^{p-q}
    \\& \Rightarrow \alpha \E\sqb{\of{\frac{X-k}{c+k}}_+^p}\leq (1-\alpha) \E\sqb{\of{\frac{X-k}{c+k}}_-^q} 
    \\& \Rightarrow \alpha \E\sqb{\of{\frac{X+c-\tilde{k}}{\tilde{k}}}_+^p}\leq (1-\alpha) \E\sqb{\of{\frac{X+c-\tilde{k}}{\tilde{k}}}_-^q},
\end{align*}
where $\tilde{k}=k+c$. 
By the definition of generalized Orlicz premia, we have
$$
H_{\Phi}(X+c)\leq \tilde{k}=H_{\Phi}(X)+c.
$$
\noindent
The case $p\leq q$ follows by a similar argument. 
\end{example}

\section{Two classes of generalized Orlicz premia}
\label{sec:gg-conv}
The flexibility gained by relaxing convexity in Definition~\ref{def:orl}, opens the door to the consideration of much more general families of Orlicz functions and, hence, induced Orlicz premia. 
In this section, we focus on two specific classes of nonconvex loss functions for which the corresponding Orlicz premia display interesting actuarial and mathematical properties.  
\smallskip

The first class consists of \emph{concave-convex} Orlicz functions, which are based on the corresponding idea of concave-convex utility functions of Friedman and Savage~\cite{FS48} and of S-shaped value functions of Kahneman and Tversky~\cite{KT79}. 
In these theories, the qualitative shape (concavity above the reference level, convexity below) is aimed to capture the behavioral pattern of risk aversion in the domain of gains and risk seeking in the loss domain. 
From the actuarial perspective, this allows for premium principles in which deviations from a reference outcome are weighted asymmetrically and according to two qualitatively distinct regimes (concave-convex or convex-concave).\smallskip

The second class consists of \emph{GA-convex} Orlicz functions, for which the corresponding Orlicz premium is geometrically convex (GG-convex). 
In the recent literature, it is demonstrated that GG-convexity is the natural counterpart of the usual (arithmetic) convexity when the underlying objects are multiplicative quantities such as gross returns (see, e.g.,~\cite{BLR18,BLR21,LR22} for the initial theory of return risk measures (RRMs), and \cite{LRZ24,LRZ23,MR22} for related connections to star-shaped risk measures). \smallskip 

Furthermore, Proposition~\ref{prop:GA-cx}, item d), below shows that GA-convex functions are actually those functions that are more convex than the logarithm,
in the sense of the relative notion of convexity based on pointwise comparison of the Arrow-Pratt coefficient (with the proper sign change for a loss function) given by $\Phi''(x)/\Phi'(x)$. 
The class of GA-convex Orlicz functions is strictly larger than the class of convex ones, and accordingly the class of GG-convex Orlicz premia strictly enlarges the convex Orlicz premia, exhibiting a parallelism between the two convexity notions that will become particularly transparent in Corollary~\ref{cor:GG-CI}.

\subsection{Concave-convex Orlicz functions}
The S-shaped value functions of prospect theory describe how individuals 
evaluate gains and losses asymmetrically with respect to a reference point: 
concave above (risk aversion in gains) and convex below (risk seeking in 
losses), with possibly different curvatures on the two sides of the reference point. 
In the Orlicz 
framework, the natural reference point is $x=1$, since $\Phi(1)\leq 1<\Phi(x)$ 
for $x>1$ separates the two regimes. 
This 
suggests building $\Phi$ by piecing together two normalized convex Young functions on 
the gain and loss sides of the reference point.
(The convex-concave case can be constructed \emph{mutatis mutandis}.)
\begin{definition}
\label{def:conc-conv-orl}
We say that $\Phi$ is a concave-convex Orlicz function if
$$
\Phi(x)= 1+ \alpha \Phi_1 \left ((x-1)_+ \right) - (1-\alpha) \Phi_2 \left ((x-1)_- \right),
$$
where $\alpha \in (0,1)$ and $\Phi_1$,$\Phi_2$ are normalized, convex and finite Young functions. 
\end{definition}
It is immediate to check that the Orlicz function $\Phi$ in this definition satisfies a), b), c) of Definition \ref{def:orl}.\ The prototypical case has been considered in Example~\ref{ex:lp-quantiles} with $p>1$ and
the corresponding Orlicz premium  $z_{\alpha,p}$ is an $L^{p+1}$-quantile in the terminology of \cite{C96}, where $p=1$ corresponds to the case of expectiles. 
By construction, $z_{\alpha,p}$ satisfies monotonicity, cash-additivity and positive homogeneity, but not convexity. 
In the case of an integer $p$, consistency properties of these quantiles with respect to $p$-convex orders have been determined in \cite {B12}. 
The following proposition collects a few general properties of Orlicz premia induced by concave-convex Orlicz functions.

\begin{proposition}\label{prop:conv-concave-orl}
Let $\Phi$ be a concave-convex Orlicz function as in Definition~\ref{def:conc-conv-orl}, with $\alpha \in (0,1)$ and $\Phi_1, \Phi_2$ finite, continuous and strictly increasing normalized convex Young functions. 
Let $X \in L^\infty_+$ with $H_\Phi(X)>0$. Then:
\begin{enumerate}[left=0pt, itemsep=0pt]
\item[a)] \emph{(First-order condition).} $k = H_\Phi(X)$ is the unique positive solution of
\begin{equation}\label{eq:foc-cc}
\alpha\,\E\!\left[\Phi_1\!\left(\tfrac{(X-k)_+}{k}\right)\right] = (1-\alpha)\,\E\!\left[\Phi_2\!\left(\tfrac{(k-X)_+}{k}\right)\right].
\end{equation}
\item[b)] \emph{(Monotonicity in the parameters).}
Denote by $H_{\alpha,\Phi_1,\Phi_2}$ the Orlicz premium associated with the parameters $(\alpha,\Phi_1,\Phi_2)$.
\begin{enumerate}[leftmargin=1.4em, itemsep=0pt]
\item[(i)] If $\Phi_1\leq \tilde\Phi_1$ pointwise, then $H_{\alpha,\Phi_1,\Phi_2}(X)\leq H_{\alpha,\tilde\Phi_1,\Phi_2}(X)$. 
\item[(ii)] If $\Phi_2\leq \tilde\Phi_2$ pointwise, then $H_{\alpha,\Phi_1,\Phi_2}(X)\geq H_{\alpha,\Phi_1,\tilde\Phi_2}(X)$.
\item[(iii)] The map $\alpha \mapsto H_{\alpha,\Phi_1,\Phi_2}(X)$ is nondecreasing on $(0,1)$, and strictly increasing whenever $P(X\neq H_{\alpha,\Phi_1,\Phi_2}(X))>0$. 
\end{enumerate}
\item[c)] \emph{($(\Phi_1,\Phi_2)$-order).}
For $X, Y \in L^\infty_+$, define 
\begin{equation*}
X \preceq_{\Phi_1,\Phi_2} Y \text { if }  
\begin{cases}
\E\!\left[\Phi_1\!\left(\tfrac{(X-t)_+}{t}\right)\right] \leq \E\!\left[\Phi_1\!\left(\tfrac{(Y-t)_+}{t}\right)\right] \\[4pt]
\E\!\left[\Phi_2\!\left(\tfrac{(t-X)_+}{t}\right)\right] \geq \E\!\left[\Phi_2\!\left(\tfrac{(t-Y)_+}{t}\right)\right]
\end{cases}
\text {for all } t>0.
\end{equation*}
Then $X \preceq_{\Phi_1,\Phi_2} Y$ implies $H_{\alpha,\Phi_1,\Phi_2}(X) \leq H_{\alpha,\Phi_1,\Phi_2}(Y)$ for every $\alpha \in (0,1)$. 
\end{enumerate}
\end{proposition}

As a corollary, we generalize to the case of $L^{p,q}$-quantiles of Example~\ref{ex:lpq-quantiles} the isotonicity results for $L^p$-quantiles found in \cite{B12}. 

\begin{corollary}\label{cor:lpq-order}
Let $\Phi_1(x) = x^p$ and $\Phi_2(x) = x^q$, with $p, q \geq 1$, so that $H_{\alpha,\Phi_1,\Phi_2}$ is the $L^{p,q}$-quantile of Example~\ref{ex:lpq-quantiles}. 
For $X, Y \in L^\infty_+$, define
\[
X \preceq_{p,q} Y \text { if } 
\begin{cases}
\E[(X-t)_+^p] \leq \E[(Y-t)_+^p] \\[2pt]
\E[(t-X)_+^q] \geq \E[(t-Y)_+^q]
\end{cases}
\text {for all } t >0. 
\]
Then:
\begin{enumerate}[left=0pt, itemsep=0pt]
\item[a)] $X \preceq_{p,q} Y$ implies $H_{\alpha,\Phi_1,\Phi_2}(X) \leq H_{\alpha,\Phi_1,\Phi_2}(Y)$ for every $\alpha \in (0,1)$. 
\item[b)] For $p = q$ even integers, the order $\preceq_{p,p}$ coincides with the $p$-convex order of \cite{B12}, and the statement in~(a) recovers the consistency result for $L^p$-quantiles therein.    
\end{enumerate}
\end{corollary}

\subsection{GA-convex Orlicz function}
The notion of GA-convexity arises in the so-called algebraic theory of convexity, in which the geometric mean (denoted by G) replaces the arithmetic mean (denoted by A) in the usual definition of convexity, leading to the following.

\begin{definition}\label{def:GA-conv}
A function $f \colon [0,+\infty) \to [-\infty, +\infty] $ is GA-convex if for each $x,y \in [0,+\infty) $ and $\lambda \in (0,1)$ it holds that
\[
f(x^\lambda  y^{1-\lambda}) \leq \lambda f(x) + (1-\lambda) f(y),
\]
where we set $- \infty + \infty = +\infty$ by definition.\ 
\end{definition}




For completeness, we recall in the following proposition the basic properties of GA-convex functions that can be found in \cite{N00, NP04}.
\begin{proposition}\label{prop:GA-cx}
Let $f\colon [0,+\infty) \to [0, +\infty)$. 
Then:
\begin{itemize}
\item [a)] If $f$ is nondecreasing and convex, then it is GA-convex
\item [b)] $f$ is GA-convex if and only if $f(e^x)$ is convex
\item [c)] if $f$ is of class $C^2$, then $f$ is GA-convex if and only if
$$
xf''(x) + f'(x) \geq 0
$$
\item [d)] if $f$ is of class $C^2$ with $f'(x)>0$, then $f$ is GA-convex if and only if
\begin{equation} 
\label{eq:comp-conv}
\frac{f''(x)}{f'(x)} \geq -\frac{1}{x}.
\end{equation}
\end{itemize}
\end{proposition}

\smallskip

The quantity on the left-hand side of equation \eqref{eq:comp-conv} is referred to as the \emph{local convexity coefficient} of $f$ at $x$, and measures the curvature of $f$ in a neighborhood of $x$, similar to the Arrow-Pratt coefficient of absolute risk aversion in utility theory. 
Equation \eqref{eq:comp-conv} thus gives an interpretation of GA-convexity as a form of comparative convexity: $f$ is GA-convex if and only if it is more convex than a log function, while $f$ is convex if and only if it is more convex than a linear function.  
A closely related notion is that of a GG-convex functional.\smallskip 

\begin{definition}\label{def:GG-GA-2}
A functional $\rho \colon L^\infty_{+} \to [0,\infty]$ is GG-convex if, for each $X,Y \in L^\infty_{+}$ 
and $\lambda \in (0,1)$, 
\begin{equation*}
\rho (X^\lambda Y^{1-\lambda}) \leq  \rho^\lambda(X)  \rho^{1-\lambda} (Y), 
\end{equation*}
where we set $0 \cdot (+\infty) = +\infty$ by definition. 
\end{definition}
\noindent

\begin{proposition} Let $\Phi$ and $\rho$ be as in Definitions~\ref{def:GA-conv} and~\ref{def:GG-GA-2}. 
Then:
\label{prop:alg-conv}
\begin{itemize}
\item [a)]If $\Phi$ is nondecreasing and convex, then it is GA-convex
\item [b)]If $\rho$ is monotone, positively homogeneous and convex, then it is GG-convex.
\end{itemize}
\end{proposition}

From Proposition~\ref{prop:orl}, items f) and g), we know that the convexity of $\Phi$ is equivalent to the convexity of $H_\Phi$, under very general assumptions on $\Phi$. 
The next proposition shows that a similar characterization applies to GG-convexity. 

\begin{proposition}\label{prop:gg-GA}
Let $\Phi$ and $H_\Phi(X)$ be as in Definition~\ref{def:orl}. 
If $\Phi$ is GA-convex, then $H_\Phi$ is GG-convex. 
Conversely, if there exists $u < 1$ with $\Phi (u) < 1$, then $H_\Phi$ is GG-convex only if $\Phi$ is GA-convex.
\end{proposition}

\smallskip

So, under general assumptions on $\Phi$, the GG-convexity of $H_\Phi$ corresponds to the GA-convexity of $\Phi$.\ This is consistent with items a) and b) in Proposition~\ref{prop:alg-conv}, which imply that the class of GG-convex generalized
Orlicz premia is strictly larger than the class of convex generalized
Orlicz premia, and that the corresponding class of nondecreasing
GA-convex Orlicz functions is strictly larger than the class of
nondecreasing convex Orlicz functions.\ 
Reconsidering the examples at the beginning of Section~\ref{sec:Orl},
it is not difficult to verify that the Orlicz function $\Phi$ is GA-convex in 
Example~\ref{ex:lpq-quantiles} in the case $\alpha \geq 1/2$ and $p=q=1$ that corresponds to convex expectiles, whereas it does not satisfy the GA-convexity property in all the other cases.\ As a corollary of Proposition~\ref{prop:gg-GA}, we obtain the following stronger formulation of item b) in Theorem~\ref{th:expectiles}. 

\begin{corollary}
\label{cor:GG-CI}
Let $\Phi$ and $H_\Phi(X)$ be as in Definition~\ref{def:orl} with $\Phi$ continuous and strictly increasing.\ 
If $H_{\Phi}$ is cash-additive and geometrically convex (resp.\ concave), then $H_\Phi(X)=e_\alpha(X)$ with $\alpha \geq 1/2$ (resp.\ $\alpha \leq 1/2$). 
\end{corollary}

Convex expectiles are thus the only geometrically convex and cash-additive Orlicz premia. 
Hence, under cash-additivity, GG-convexity of $H_\Phi$ is equivalent to its convexity, and the corresponding Orlicz premia are exactly the convex expectiles. 
The two classes diverge only when cash-additivity is dropped.

\noindent

\section{Axiomatization of generalized Orlicz premia}
\label{sec:CxLS}
As is well-known, law-invariant functionals such as generalized Orlicz premia can be seen as functionals defined on sets of distribution functions.
An important property that these functionals can have, is the convexity of their level sets with respect to mixtures, referred to in the literature for short as the CxLS property.

\begin{definition}\label{def:CxLS}
Let $\M' \subseteq \M$ be a convex set of distribution functions of probability measures on $\R$. 
A functional $\rho \colon \M' \to \R$ has the CxLS property if
\[
\rho(F)=\rho(G)=\gamma \Rightarrow \rho(\lambda F +(1-\lambda)G)=\gamma,
\]
for each $\gamma \in \R$, $F,G \in \M'$ and $\lambda \in (0,1)$.
\end{definition}
\noindent 
The CxLS property is a necessary condition for the elicitability of a functional, which can be informally defined as the property of being the minimizer of a suitable expected loss function. 
The relevance of this notion has been thoroughly discussed in the financial and statistics literature; we refer to e.g.,\  \cite{G11,BB15,Z16,EMWW21}.\smallskip

A stream of literature has studied elicitable coherent or convex risk measures, by characterizing coherent or convex risk measures satisfying the CxLS property; see, e.g.,\ \cite{W06, BB15, Z16, DBBZ14}.\ Most of the attention has been given to the cash-additive case, with the exception of \cite{BLR18}, in which the positively homogeneous case has been considered, leading in Theorem~2 of \cite{BLR18} to an axiomatization of Orlicz premia that we sharpen in the following.\smallskip  

First, we verify in the following lemma that generalized Orlicz premia satisfy the CxLS property. 
\begin{lemma}
\label{lem:cxls}
Let $\Phi$ and $H_\Phi$ be as in Definition~\ref{def:orl}.\ 
Assume that $H_\Phi(X)=H_\Phi(Y) = \gamma \in [0, + \infty)$, with $X \sim F$ and $Y \sim G$. 
Then, for each $\lambda \in (0,1)$, 
$$
Z \sim \lambda F + (1-\lambda) G \Rightarrow H_\Phi(Z) = \gamma. 
$$
\end{lemma}

As a preliminary lemma that will be useful in the proof of Theorem~\ref{th:axiom}, we show that the CxLS property implies the convexity with respect to mixtures of the acceptance set at the level of distributions as defined below.  

\begin{lemma}\label{lem:mconv}
Let $\rho \colon L^\infty_+ \to [0,+\infty)$ be monotone, positively homogeneous, law-invariant and positive.\ If $\rho$ has the CxLS property, then the multiplicative acceptance set at the level of distributions defined by
\[
B_\rho := \{ F \in \M_{1,c}([0,+\infty)) \mid \rho(F) \leq 1 \}
\]
and its complement $B_\rho^c$ are both convex with respect to mixtures.
\end{lemma}

Lemma~\ref{lem:mconv} is a multiplicative counterpart of Lemma~3.1(a) in \cite{DBBZ14}, where the analogous result is established for monetary law-determined utilities. 
\smallskip

In order to formulate the topological assumption needed for Theorem~\ref{th:axiom}, we recall from \cite{FS16} the notions of $\psi$-weak convergence and $\psi$-weak semicontinuity. 

\begin{definition}\label{def:psi-weak}
A continuous function $\psi \colon [0,+\infty) \to [1,+\infty)$ is called a \emph{gauge function}. 
The associated vector space of test functions is 
\[
C^{\psi} := \Big\{ f \colon [0,+\infty) \to \R \ \text{continuous} \ \Big| \ \exists\, c \geq 0 \ \text{s.t.} \ |f(x)| \leq c\, \psi(x) \ \text{for all } x \geq 0 \Big\}. 
\]
The \emph{$\psi$-weak topology} on $\M_{1,c}([0,+\infty))$ is the coarsest topology that makes the maps $F \mapsto \int f\, dF$ continuous for every $f \in C^{\psi}$.\ We write $F_n \xrightarrow{\psi} F$ for convergence in this topology. 
A functional $\rho \colon \M_{1,c}([0,+\infty)) \to \R$ is said to be \emph{$\psi$-weakly continuous} if 
$
F_n \xrightarrow{\psi} F \ \Longrightarrow \ \rho(F_n) \to \rho(F) 
$
and \emph{$\psi$-weakly lower semicontinuous} if 
$
F_n \xrightarrow{\psi} F \ \Longrightarrow \liminf \rho(F_n) \geq \rho(F). 
$
\end{definition}

The next lemma clarifies the relationship between $\psi$-weak lower semicontinuity 
and the closure of the multiplicative acceptance set at the level of distributions, 
and shows that the topological requirement of Theorem~\ref{th:axiom} below is 
automatically satisfied by every generalized Orlicz premium. Further, it identifies exactly for which Orlicz functions the Orlicz premium is continuous from above. 

\begin{lemma}\label{lem:psi-lsc}
Let $\rho \colon \M_{1,c}([0,+\infty)) \to [0,+\infty)$ be a law-invariant 
functional, and let $\psi$ be a gauge function. 
Then:
\begin{enumerate}[label=\alph*), left=0pt, itemsep=0pt]
\item $\rho$ is $\psi$-weakly lower semicontinuous if and only if, for every 
$\gamma \in [0,+\infty)$, the sublevel set 
$\{F \in \M_{1,c}([0,+\infty)) \mid \rho(F) \leq \gamma\}$ is $\psi$-weakly closed.\ 
In particular, $\psi$-weak lower semicontinuity of $\rho$ implies that the 
multiplicative acceptance set $B_\rho$ is $\psi$-weakly closed.
\item Every generalized Orlicz premium $H_\Phi$ in the sense of 
Definition~\ref{def:orl} is weakly lower semicontinuous and hence $\psi$-weakly lower semicontinuous for every gauge 
function $\psi$.

\item Assume in addition that $\Phi$ is right-continuous at $0$ and $\Phi(+\infty)=+\infty$. Then $H_\Phi$ is continuous 
from above in the sense that 
$$
X_n \in L^\infty_+, X_n \downarrow X \Rightarrow H_\Phi(X_n) \to H_\Phi(X)
$$
if and only if $\Phi$ is strictly increasing on 
$\{x \geq 0 \mid \Phi(x)>0\}$.
\end{enumerate}
\end{lemma}

\smallskip

Lemma~\ref{lem:psi-lsc}\,a) is the multiplicative analogue of the equivalence used 
in \cite{W06} between $\psi$-weak lower semicontinuity of the underlying functional 
and $\psi$-weak closedness of the acceptance set in the cash-additive setting; 
the lower semicontinuity assumption of 
Theorem~\ref{th:axiom} below does not exclude any generalized Orlicz premium and 
in fact holds in the strongest form (with respect to the constant gauge 
$\psi \equiv 1$, i.e.,\ ordinary weak convergence). 
\smallskip

The following theorem is the main result of this section. It refines the 
axiomatization of Orlicz premia obtained in Theorem~2 of \cite{BLR18}, where 
the return risk measure was assumed to satisfy directly condition (3.1) of 
\cite{W06} and the characterization was established on $L^\infty_{++}$. 
Here we show that, once continuity from above is imposed, condition (3.1) needs 
no longer be assumed but follows from the other axioms, and the resulting 
identity $\rho=H_\Phi$ extends to the whole $L^\infty_+$. 

\begin{theorem}\label{th:axiom}
A functional $\rho \colon L_{+}^{\infty}\to [0,+\infty)$ is law-invariant, 
monotone, positive, positively homogeneous, normalized, and satisfies the 
CxLS property, weak lower semicontinuity and continuity from above if and 
only if there exists a generalized Orlicz function $\Phi$ in the sense of 
Definition~\ref{def:orl}, right-continuous at $0$, strictly increasing on 
$\{x \geq 0 : \Phi(x)>0\}$ and with $\Phi(+\infty)=+\infty$, such that 
$$\rho(X) = H_\Phi(X), \text{ for all } X \in L_+^\infty.$$
\end{theorem}
\smallskip 

\begin{remark}
Theorem~\ref{th:axiom} provides an if-and-only-if identification only for a 
proper subclass of the generalized Orlicz premia of Definition~\ref{def:orl}, 
namely those whose representing function $\Phi$ is strictly increasing on 
$\{\Phi>0\}$, right-continuous at $0$, and satisfies $\Phi(+\infty)=+\infty$. 
These restrictions 
guarantee  continuity 
from above and condition (3.1) of \cite{W06}.
Among the positive 
generalized Orlicz premia, the only one not covered by the theorem 
is the essential supremum, corresponding to the case $\Phi(0)=1$ of 
item~a) of Lemma~\ref{lem:positivity}. 
\end{remark}

Recalling that the CxLS property is a necessary condition for elicitability, we thus find that generalized Orlicz premia naturally arise as the only elicitable premium principles satisfying the assumptions of Theorem~\ref{th:axiom}. 
The study of families of loss functions consistent with Orlicz premia and their actuarial and financial applications is being pursued in a separate paper. 

\section{Concluding remarks}\label{sec:con}
Our analysis has shown that nonconvex Orlicz premia constitute a natural and 
far-reaching generalization of convex Orlicz premia.\ Our extended definition 
reveals a connection with various functionals, such as expectiles and other 
generalized quantiles, that has so far not been recognized in the literature. 
We have established that cash-additive generalized Orlicz premia do not collapse 
to the mean, contrary to conventional wisdom, but instead give rise to 
$L^{p}$-quantiles.\smallskip 

Further, relaxing convexity naturally leads to the consideration of 
concave-convex loss functions, or loss functions that are more convex than a 
reference loss function, as in the case of GA-convex Orlicz functions. 
Finally, we have shown that generalized Orlicz premia admit a complete axiomatic 
characterization: a natural subclass of them coincides with the law-invariant, 
monotone, positively homogeneous, positive, normalized functionals that are 
weakly lower semicontinuous, continuous from above, and enjoy the CxLS property. 
Since the CxLS property is necessary for elicitability, generalized Orlicz premia 
thus provide the natural framework for applying loss-based statistical methods to actuarial premium calculation. 

We believe that the generalized definition of Orlicz premia introduced in this 
paper should be adopted as the new canonical one.

\appendix
\renewcommand{\theequation}{\thesection.\Roman{equation}}
\setcounter{equation}{0}

\section{Proofs}

\begin{proof}[Proof of Lemma~\ref{lem:equiv-orl}]
Sufficiency of \eqref{eq:affine-equiv} is immediate, since $c>0$ gives $\E[\Phi_2(X/k)]\leq 1\iff\E[\Phi_1(X/k)]\leq 1$. 
For necessity, Proposition~\ref{prop:orl}, item c), below yields $H_{\Phi_i}(X)\leq 1\iff\E[\Phi_i(X)]\leq 1$, so $H_{\Phi_1}=H_{\Phi_2}$ implies that the affine functionals $L_i(\mu):=\int\Phi_i\,d\mu$ share the same sub-level set $\{L_i\leq 1\}$ on the convex set of laws of elements of $L^\infty_+$. 
Fixing $\nu$ with $b:=L_1(\nu)>1$ and writing $d:=L_2(\nu)$, every $\mu$ with $a:=L_1(\mu)<1$ admits a unique $\lambda^*=(b-1)/(b-a)\in(0,1)$ for which $\lambda^*\mu+(1-\lambda^*)\nu$ lies on the common boundary, forcing $\lambda^*L_2(\mu)+(1-\lambda^*)d=1$, and hence
\[
L_2(\mu)=c\,L_1(\mu)+(1-c),\qquad c=\frac{d-1}{b-1}>0.
\]
Specializing to $\mu=\delta_x$ yields the pointwise affine relation \eqref{eq:affine-equiv}.
\end{proof}
\smallskip

\begin{proof}[Proof of Lemma~\ref{lemma:vec}]
It is trivial to show closure of $L^\Phi$ with respect to 
scalar multiplication. 
Since $\Phi$ is concave with $\Phi(0)=0$, it follows that $\Phi(x)/x$ is nonincreasing, so 
$$
\frac{\Phi(x)}{x} \geq \frac{\Phi(x+y)}{x+y}, \quad
\frac{\Phi(y)}{y} \geq \frac{\Phi(x+y)}{x+y},
$$
which gives 
\[
\Phi(x+y) 
= \frac{x}{x+y}\,\Phi(x+y) + \frac{y}{x+y}\,\Phi(x+y) 
\leq \Phi(x) + \Phi(y).
\]
The inequality extends to $x=0$ or $y=0$ by $\Phi(0)=0$ and 
monotonicity. \\\smallskip
Let now $X, Y \in L^\Phi$, so there exist 
$k_X, k_Y > 0$ with $\E[\Phi(|X|/k_X)] \leq 1$ and 
$\E[\Phi(|Y|/k_Y)] \leq 1$. 
Set $k_0 = k_X + k_Y$. 
Using $|X+Y| \leq |X|+|Y|$, monotonicity of $\Phi$, subadditivity 
and $|X|/k_0 \leq |X|/k_X$ (analogously for $Y$) we get:
\[
\Phi\!\left(\frac{|X+Y|}{k_0}\right) 
\leq \Phi\!\left(\frac{|X|+|Y|}{k_0}\right) 
\leq \Phi\!\left(\frac{|X|}{k_0}\right) + \Phi\!\left(\frac{|Y|}{k_0}\right) 
\leq \Phi\!\left(\frac{|X|}{k_X}\right) + \Phi\!\left(\frac{|Y|}{k_Y}\right).
\]
Taking expectations yields
\[
\E\!\left[\Phi\!\left(\frac{|X+Y|}{k_0}\right)\right] \leq 2.
\]
To conclude, we show that $\E[\Phi(|X+Y|/k)] \leq 1$ for some 
$k > 0$. 
For every $k \geq k_0$, monotonicity of $\Phi$ gives 
$0 \leq \Phi(|X+Y|/k) \leq \Phi(|X+Y|/k_0)$, and the latter is 
integrable. 
Moreover, by right-continuity of $\Phi$ at $0$ and 
$\Phi(0)=0$,
\[
\Phi(|X+Y|/k) \to 0 \quad \text{pointwise as } k \to \infty.
\]
By monotone convergence, 
$\E[\Phi(|X+Y|/k)] \to 0$ as $k \to \infty$, so there exists 
$k^* > 0$ with $\E[\Phi(|X+Y|/k^*)] \leq 1$. 
Hence, 
$X+Y \in L^\Phi$.
\end{proof}
\smallskip

\begin{proof}[Proof of Proposition~\ref{prop:orl}]
\emph{a)} From the assumptions on $\Phi$, for $X \in L^\Phi$, the set 
$I_X:=\{ k>0 \mid \E[\Phi(X/k)] \leq 1 \}$ is a nonempty unbounded interval. 
Monotonicity of $H_\Phi$ follows from the monotonicity of $\Phi$ and the proof of positive homogeneity of $H_\Phi$ is standard. 
Notice that $H_\Phi(0)=0$ since  $\Phi(0) \leq 1$. 
Since $\E[\Phi(1/k)] \leq 1$ for $k \geq 1$ and $\E[\Phi(1/k)] > 1$ for $k<1$, we have $H_\Phi(1)=1$. 

\emph{b)} Let $g(k):=\E \left[\Phi(X/k) \right]$.\ Let $k \geq H_\Phi(X)$.
If $k_n \downarrow k$, then, from left-continuity of $\Phi$, it follows that $\Phi(X/k_n) \uparrow \Phi(X/k)$.\ 
From monotone convergence it follows that $g$ is right-continuous, so 
$
H_\Phi(X) = \inf \{ k \mid g(k) \leq 1 \} = \min \{ k \mid g(k) \leq 1 \}. 
$

\emph{c)} The `if' part follows immediately from Definition~\ref{def:orl}. 
For the `only if' part, we distinguish between two cases. 
If $H_\Phi(X)>0$, it follows from b). 
If $H_\Phi(X)=0$, there exist $k_n \downarrow 0$ such that $\E[\Phi(X/k_n)] \leq 1$, so, from monotonicity, $\E[\Phi(X)] \leq 1$.  

\emph{d)} Under these assumptions, the function $g(k)=\E \left[\Phi(X/k) \right]$ introduced in the proof of b) is continuous and strictly decreasing, from which the thesis follows. 

\emph{e)} The proof is standard. 

\emph{f)} Assume by contradiction that $\Phi$ is not midconvex, i.e., there exist $x_1,x_2 \in (0, +\infty)$ such that $$
b:=\Phi \left ((x_1+x_2)/2 \right )>(\Phi(x_1)+\Phi(x_2))/2:=a.
$$
We want to prove that 
there exist $z \in (0,+\infty)$ and $c=\Phi(z)$ such that
\begin{equation}\label{eq:conv-1}
\begin{cases}
\lambda c +(1-\lambda)b > 1 \\
\lambda c +(1-\lambda)a \leq 1
\end{cases}   
\end{equation}
for some $\lambda \in (0,1)$, or equivalently that 
\[
c \in I_\lambda := 
\left ( \frac{1-b(1-\lambda)}{\lambda}, \frac{1-a(1-\lambda)}{\lambda} \right ].
\]
There are three cases. 
If $a\leq 1<b$, then $c=1$ satisfies \eqref{eq:conv-1} for each $\lambda \in (0,1)$. 
If $a < b \leq 1$, then $z>1$ and $c=\Phi(z)$ satisfies \eqref{eq:conv-1} since any $c>1$ is contained in some $I_\lambda$, while if $1 \leq a < b$, then $z=u$ and $c=\Phi(u)$ satisfies \eqref{eq:conv-1} since any $c<1$ is contained in some $I_\lambda$. 
As a consequence, 
\begin{equation}\label{eq:convexity}
\lambda \Phi(z)+(1-\lambda)\Phi\left ((x_1+x_2)/2 \right )>1 \geq \lambda \Phi(z)+(1-\lambda) \frac{\Phi(x_1)+\Phi(x_2)}{2}.
\end{equation}
Let $A,B,C \in \F$ be disjoint sets with $P(A)=\lambda$, $P(B)=P(C)=\frac{1-\lambda}{2}$ and let
\begin{align*}
X&=z 1_A+x_1 1_B+ x_2 1_C\\
Y&=z 1_A+x_2 1_B+ x_1 1_C\\
Z&=z 1_A+\frac{x_1+x_2}{2} 1_{B\cup C}=\frac{X+Y}{2}.
\end{align*}
From \eqref{eq:convexity}, we have $\E[\Phi(Z)]>1 \Rightarrow H_\Phi(Z) > 1$ and $\E[\Phi(X)] \leq 1 \Rightarrow H_\Phi (X) \leq 1$, $\E[\Phi(Y)] \leq 1 \Rightarrow H_\Phi (Y) \leq 1$, which contradicts the convexity of $H_{\Phi}$. 
As a consequence, $\Phi$ is midconvex, and since it is also nondecreasing, it is convex.

\emph{g)} Follows immediately from the definition. 
\end{proof}
\smallskip

\begin{proof}[Proof of Lemma~\ref{lem:positivity}]
Let $X\in L^\infty_+$ with $p:=P(X>0)\in(0,1]$. 
By monotone 
convergence,
\[
\lim_{k\downarrow 0}\E[\Phi(X/k)]
=\Phi(+\infty)\,p+\Phi(0)\,(1-p).
\]
\emph{a)} If $\Phi(0)=1$, then $\Phi\equiv 1$ on $[0,1]$ by item~a) of its definition
and monotonicity, while $\Phi>1$ on $(1,+\infty)$. 
Hence, 
$\E[\Phi(X/k)]\geq 1$ for every $k>0$, with equality if and only if 
$X/k\leq 1$ a.s., i.e., if $k\geq\esssup X$. 
This gives 
$H_\Phi(X)=\esssup X$, which is strictly positive whenever 
$P(X>0)>0$. 

\emph{b)} If $\Phi(0)<1$ and $\Phi(+\infty)=+\infty$, the limit 
above is $+\infty$, so $\E[\Phi(X/k_0)]>1$ for some $k_0>0$, 
yielding $H_\Phi(X)\geq k_0>0$.

\emph{c)} If $\Phi(0)<1$ and $M:=\Phi(+\infty)<+\infty$, take  $X:=\mathbf 1_A$ with 
$P(A)=p \leq (1-\Phi(0))/(M-\Phi(0))$. 
Then, for every $k>0$,
\[
\E[\Phi(X/k)]\leq p\,M+(1-p)\,\Phi(0)\leq 1,
\]
so $H_\Phi(X)=0$ with $P(X>0)=p>0$, showing that $H_\Phi$ is not positive. 
\end{proof}
\smallskip

\begin{proof}[Proof of Theorem~\ref{th:expectiles}]
For the `if' part of a), in Example \ref{ex:lp-quantiles} it has been shown that in this case $H_\Phi$ is cash-additive. 
For the `if' part of b), in Example \ref{ex:expectiles} it has been shown that $H_\Phi$ is cash-additive and convex (concave) if $\alpha \geq (\leq) 1/2$. \\
We now prove the `only if' part. 
Let $x_1 < 1$ and $x_2>1$. 
From the assumptions on $\Phi$, there exists $p$ such that
$
p \Phi(x_1) +(1-p) \Phi(x_2) = 1,
$
and from Proposition~\ref{prop:orl} it follows that $H_\Phi(X)=1$, where 
$X=x_1 1_A + x_2 1_{A^c}$ with $P(A)=p$.\ 
From positive homogeneity and cash-additivity, it follows that 
$
H_\Phi(cX - c + 1) = c - c + 1 = 1 \text { for any } c \geq 0,
$
so again from Proposition~\ref{prop:orl} it follows that 
$
p \Phi(c(x_1-1)+1) +(1-p) \Phi(c(x_2-1)+1) = 1, \text { for any } c \geq 0. 
$
Let $u:=x_1-1 < 0$ and $v:=x_2-1>0$. 
The previous arguments have shown that if
\begin{equation}\label{eq:p}
pf(u)+(1-p)g(v)=0,
\end{equation}
then 
\begin{equation*}
pf(cu)+(1-p)g(cv)=0,  \text { for any } c \geq 0, 
\end{equation*}
where $f$ and $g$ denote the restrictions of the function $h(x):=\Phi(x+1)-1$ to the domains $(-1,0)$ and $(0,+ \infty)$, respectively. 
Since from the assumptions on $\Phi$ it follows that for any $u \in (-1,0)$ and $v \in (0,+\infty)$ it is possible to find $p$ such that \eqref{eq:p} holds, we find that for each $u$ and $v$, 
\[
\frac{f(u)}{g(v)} = \frac{f(cu)}{g(cv)}, \text { for any } c > 0, 
\]
which can be recast into a multiplicative Pexider functional equation (see, e.g.,\ \cite{A66}, Theorem~4 in Section~3.1), whose general solution is
\[
\begin{cases}
f(u)= -a (-u)^p &\text{if } u \in (-1,0),\\[2pt]
g(v)= b v^p &\text{if } v \in (0,+\infty),
\end{cases}
\]
with $a>0$, $b>0$ and $p\geq 0$. 
The case $p=0$ would force $f$ and $g$ to
be constant on their respective domains, contradicting the strict
monotonicity of $\Phi$; hence $p>0$. 
Going back to $\Phi$ via $\Phi(x)=1+h(x)$
we obtain
\begin{equation}\label{eq:Phi-prenorm}
\Phi(x)\;=\;1\;+\;b\,(x-1)_+^p\;-\;a\,(x-1)_-^p,
\end{equation}
with $a=1-\Phi(0)\in(0,1]$ and $b=\Phi(2)-1>0$.
The rescaling
\[
\Phi\;\longmapsto\;\widetilde\Phi\;:=\;1+\tfrac{1}{a+b}\,(\Phi-1)
\]
as in Lemma \ref{lem:equiv-orl} leaves $H_\Phi$ unchanged.
Replacing $\Phi$ with $\widetilde\Phi$ amounts to imposing $a+b=1$ in
\eqref{eq:Phi-prenorm}, and setting $\alpha:=b\in(0,1)$ we obtain a).
To prove b), notice that convexity or concavity holds if and only if $p=1$, and is determined by the inequality between $a$ and $b$. 
\end{proof}
\smallskip

\begin{proof}[Proof of Proposition~\ref{prop:conv-concave-orl}]
\emph{a)} Under the assumptions, $\Phi$ is finite, continuous and strictly increasing on $[0,+\infty)$, hence d) of Proposition~\ref{prop:orl} applies and $k = H_\Phi(X)$ is the unique solution of $\E[\Phi(X/k)] = 1$. 
Expanding,
\[
\E[\Phi(X/k)] - 1 = \alpha\,\E\!\left[\Phi_1\!\left(\tfrac{(X-k)_+}{k}\right)\right] - (1-\alpha)\,\E\!\left[\Phi_2\!\left(\tfrac{(k-X)_+}{k}\right)\right],
\]
which gives~\eqref{eq:foc-cc}.

\smallskip
\emph{b)} Introduce the auxiliary function 
\[
G_{\alpha,\Phi_1,\Phi_2}(X,t) := \alpha\,\E\!\left[\Phi_1\!\left(\tfrac{(X-t)_+}{t}\right)\right] - (1-\alpha)\,\E\!\left[\Phi_2\!\left(\tfrac{(t-X)_+}{t}\right)\right], \qquad t>0,
\] 
so that~\eqref{eq:foc-cc} reads $G_{\alpha,\Phi_1,\Phi_2}(X,k)=0$. For fixed $X$ and parameters, $t\mapsto G_{\alpha,\Phi_1,\Phi_2}(X,t)$ is strictly decreasing, since $t\mapsto (X-t)_+/t$ is nonincreasing, $t\mapsto (t-X)_+/t$ is nondecreasing, and $\Phi_1, \Phi_2$ are strictly increasing. 
Hence, writing $k = H_{\alpha,\Phi_1,\Phi_2}(X)$ and $\tilde k = H_{\alpha,\tilde\Phi_1,\tilde\Phi_2}(X)$, we have $k \leq \tilde k$ if and only if $G_{\alpha,\tilde\Phi_1,\tilde\Phi_2}(X,k) \geq 0$.

(i) If $\Phi_1\leq \tilde\Phi_1$, then at $t=k$,
\[
G_{\alpha,\tilde\Phi_1,\Phi_2}(X,k) \geq G_{\alpha,\Phi_1,\Phi_2}(X,k) = 0,
\]
yielding $H_{\alpha,\Phi_1,\Phi_2}(X)\leq H_{\alpha,\tilde\Phi_1,\Phi_2}(X)$. 
Symmetrically, (ii) follows by the opposite inequality on the $\Phi_2$-term.

(iii) For $\alpha < \tilde\alpha$, at $t=k$,
\[
G_{\tilde\alpha,\Phi_1,\Phi_2}(X,k) - G_{\alpha,\Phi_1,\Phi_2}(X,k) = (\tilde\alpha-\alpha)\left(\E\!\left[\Phi_1\!\left(\tfrac{(X-k)_+}{k}\right)\right] + \E\!\left[\Phi_2\!\left(\tfrac{(k-X)_+}{k}\right)\right]\right) \geq 0,
\]
and the inequality is strict whenever $P(X\neq k)>0$, since $\Phi_1, \Phi_2$ are strictly increasing with $\Phi_i(0)=0$. 
Therefore, $G_{\tilde\alpha,\Phi_1,\Phi_2}(X,k) \geq 0$, i.e., $H_{\tilde\alpha,\Phi_1,\Phi_2}(X)\geq k$, with strict inequality under the nondegeneracy condition.

\smallskip
\emph{c)} By the assumption $X\preceq_{\Phi_1,\Phi_2} Y$, at $t=k:=H_{\alpha,\Phi_1,\Phi_2}(X)$,
\[
\E\!\left[\Phi_1\!\left(\tfrac{(Y-k)_+}{k}\right)\right] \geq \E\!\left[\Phi_1\!\left(\tfrac{(X-k)_+}{k}\right)\right], \qquad 
\E\!\left[\Phi_2\!\left(\tfrac{(k-Y)_+}{k}\right)\right] \leq \E\!\left[\Phi_2\!\left(\tfrac{(k-X)_+}{k}\right)\right].
\]
Multiplying by $\alpha>0$ and $-(1-\alpha)<0$ and summing, 
\[
G_{\alpha,\Phi_1,\Phi_2}(Y,k) \geq G_{\alpha,\Phi_1,\Phi_2}(X,k) = 0,
\]
which, by monotonicity of $G_{\alpha,\Phi_1,\Phi_2}(Y,\cdot)$, gives $H_{\alpha,\Phi_1,\Phi_2}(Y) \geq k$. 
\end{proof}
\smallskip

\begin{proof}[Proof of Corollary~\ref{cor:lpq-order}]
\emph{a)} Follows directly from~c) of Proposition~\ref{prop:conv-concave-orl} applied to $\Phi_1(x)=x^p$ and $\Phi_2(x)=x^q$, using homogeneity of the power functions to absorb the factor $1/t$: $\Phi_1((x-t)_+/t) = t^{-p}(x-t)_+^p$, and analogously for $\Phi_2$, so the inequalities in the definition of $\preceq_{\Phi_1,\Phi_2}$ at level $t$ reduce to those in the definition of $\preceq_{p,q}$.

\emph{b)} For $p=q$, the two conditions
\[
\E[(X-t)_+^p] \leq \E[(Y-t)_+^p], \qquad \E[(t-X)_+^p] \geq \E[(t-Y)_+^p],
\]
are exactly the two families of inequalities defining the $p$-convex order in~\cite{B12}. 
\end{proof}
\smallskip

\begin{proof}[Proof of Proposition~\ref{prop:GA-cx}]
\textit{a)} By the AM-GM inequality, for every $x,y \in [0,+\infty)$ and 
$\lambda \in (0,1)$ we have
$
x^\lambda y^{1-\lambda} \leq \lambda x + (1-\lambda) y.
$
Since $f$ is nondecreasing, 
$f(x^\lambda y^{1-\lambda}) \leq f(\lambda x + (1-\lambda) y)$; 
since $f$ is convex,
$f(\lambda x + (1-\lambda) y) \leq \lambda f(x) + (1-\lambda) f(y)$.
Chaining the two inequalities yields the GA-convexity of $f$.

\textit{b)} Define $g\colon \R \to (-\infty,+\infty)$ by $g(u) := f(e^u)$.

\noindent ($\Rightarrow$) Let $u,v \in \R$, $\lambda \in (0,1)$, and set 
$x = e^u, y = e^v \in (0,+\infty)$. 
Since $e^{\lambda u + (1-\lambda) v} = x^\lambda y^{1-\lambda}$, the GA-convexity of $f$ gives
$$
g(\lambda u + (1-\lambda) v) 
= f(x^\lambda y^{1-\lambda}) 
\leq \lambda f(x) + (1-\lambda) f(y) 
= \lambda g(u) + (1-\lambda) g(v),
$$
hence $g$ is convex on $\R$.

\noindent ($\Leftarrow$) Conversely, assume $g$ convex on $\R$. 
For $x,y \in (0,+\infty)$ and $\lambda \in (0,1)$, set $u = \log x$, $v = \log y$. 
Convexity of $g$ gives
\[
f(x^\lambda y^{1-\lambda}) 
= g(\lambda u + (1-\lambda) v) 
\leq \lambda g(u) + (1-\lambda) g(v) 
= \lambda f(x) + (1-\lambda) f(y),
\]
which is the GA-convexity of $f$ on $(0,+\infty)$.

\textit{c)} By b), $f$ is GA-convex if and only if $g(u) := f(e^u)$ is convex on $\R$. 
Since $f \in C^2[0,+\infty)$, also $g \in C^2(\R)$, and convexity of $g$ is equivalent to 
$g''(u) \geq 0$ for every $u \in \R$. 
A direct computation gives
\[
g'(u) = e^u f'(e^u), \,
g''(u) = e^u f'(e^u) + e^{2u} f''(e^u).
\]
Setting $x = e^u > 0$, this can be rewritten as
\[
g''(u) = x f'(x) + x^2 f''(x) = x \bigl(f'(x) + x f''(x)\bigr).
\]
Since $x>0$, the condition $g''(u) \geq 0$ for all $u \in \R$ is equivalent to 
$x f''(x) + f'(x) \geq 0$ for all $x \in (0,+\infty)$, 
and extends to $x=0$ by continuity of $f'$ and $f''$ on $[0,+\infty)$.
\textit{d)} Since $f'(x) > 0$ for every $x \in (0,+\infty)$, 
dividing the inequality $x f''(x) + f'(x) \geq 0$ established in c) by $x f'(x) > 0$ 
yields, for $x>0$,
\[
\frac{f''(x)}{f'(x)} + \frac{1}{x} \geq 0,
\]
which is exactly~\eqref{eq:comp-conv}. 
The converse implication is obtained by 
multiplying \eqref{eq:comp-conv} by $x f'(x)>0$ and applying c).
\end{proof}
\smallskip

\begin{proof}[Proof of Proposition~\ref{prop:alg-conv}]
\textit{a)} The argument is the same as in Proposition~\ref{prop:GA-cx}\,a):
the AM--GM inequality $x^\lambda y^{1-\lambda} \le \lambda x + (1-\lambda)y$, 
together with the monotonicity and the convexity of $\Phi$, yields
$\Phi(x^\lambda y^{1-\lambda}) \le \Phi(\lambda x + (1-\lambda) y) \le \lambda \Phi(x) + (1-\lambda)\Phi(y)$
for all $x,y \in [0,+\infty)$ and $\lambda \in (0,1)$.

\textit{b)} Let $X, Y \in L^\infty_+$ and $\lambda \in (0,1)$.
If $\rho(X) = +\infty$ or $\rho(Y) = +\infty$, then by the convention 
$0\cdot(+\infty)=+\infty$ adopted in Definition~\ref{def:GG-GA-2} we have 
$\rho(X)^\lambda \rho(Y)^{1-\lambda} = +\infty$, so the GG-convexity inequality 
is trivially satisfied. 
Assume henceforth $\rho(X), \rho(Y) < +\infty$.

For every $\alpha, \beta > 0$, the pointwise AM--GM inequality applied to 
$X(\omega)/\alpha \ge 0$ and $Y(\omega)/\beta \ge 0$ gives
\[
\left(\frac{X}{\alpha}\right)^{\!\lambda} \!\left(\frac{Y}{\beta}\right)^{\!1-\lambda}
\;\le\; \lambda \, \frac{X}{\alpha} + (1-\lambda) \, \frac{Y}{\beta} 
\qquad P\text{-a.s.,}
\]
or equivalently 
$X^\lambda Y^{1-\lambda} \le \alpha^\lambda \beta^{1-\lambda}\!\left(\lambda X/\alpha + (1-\lambda) Y/\beta\right)$.
Using monotonicity, then positive homogeneity (extracting the constant $\alpha^\lambda \beta^{1-\lambda}>0$), 
then convexity (with weights $\lambda, 1-\lambda$ that sum to $1$), and finally positive homogeneity once more,
\begin{align*}
\rho(X^\lambda Y^{1-\lambda}) 
&\le \rho\!\left(\alpha^\lambda \beta^{1-\lambda}\left(\lambda \tfrac{X}{\alpha} + (1-\lambda)\tfrac{Y}{\beta}\right)\right)\\
&= \alpha^\lambda \beta^{1-\lambda} \, \rho\!\left(\lambda \tfrac{X}{\alpha} + (1-\lambda)\tfrac{Y}{\beta}\right) \\
&\le \alpha^\lambda \beta^{1-\lambda} \!\left(\lambda \, \rho\!\left(\tfrac{X}{\alpha}\right) + (1-\lambda)\,\rho\!\left(\tfrac{Y}{\beta}\right)\right) \\
&= \lambda \, \alpha^{\lambda-1}\beta^{1-\lambda} \rho(X) + (1-\lambda)\, \alpha^\lambda \beta^{-\lambda} \rho(Y).
\end{align*}
We have therefore proved that, for every $\alpha,\beta>0$,
\begin{equation} \label{eq:gg-bound}
\rho(X^\lambda Y^{1-\lambda}) 
\;\le\; \lambda \, \alpha^{\lambda-1}\beta^{1-\lambda} \rho(X) + (1-\lambda)\, \alpha^\lambda \beta^{-\lambda} \rho(Y).
\end{equation}

\noindent\emph{Case 1: $\rho(X)>0$ and $\rho(Y)>0$.} 
Setting $\alpha = \rho(X)$ and $\beta = \rho(Y)$ in \eqref{eq:gg-bound} yields 
\[
\rho(X^\lambda Y^{1-\lambda}) 
\;\le\; \lambda \, \rho(X)^\lambda \rho(Y)^{1-\lambda} + (1-\lambda)\, \rho(X)^\lambda \rho(Y)^{1-\lambda} 
\;=\; \rho(X)^\lambda \rho(Y)^{1-\lambda}.
\]

\noindent\emph{Case 2: $\rho(X) = 0$ (the case $\rho(Y) = 0$ is symmetric).} 
Setting $\beta = 1$ in \eqref{eq:gg-bound}, the first summand on the right-hand side is identically zero (because $\rho(X) = 0$) and the second is $(1-\lambda)\,\alpha^\lambda \rho(Y)$. 
Letting $\alpha \downarrow 0$ gives 
$\rho(X^\lambda Y^{1-\lambda}) \le 0$ and hence $\rho(X^\lambda Y^{1-\lambda}) = 0 = \rho(X)^\lambda \rho(Y)^{1-\lambda}$.

In both cases $\rho(X^\lambda Y^{1-\lambda}) \le \rho(X)^\lambda \rho(Y)^{1-\lambda}$, so $\rho$ is GG-convex.
\end{proof}
\smallskip

\begin{proof}[Proof of Lemma~\ref{lem:cxls}]
From the monotonicity of $\Phi$, 
$\E[\Phi(W/k)] \leq 1$  for every $k > H_\Phi(W)$ and by definition $\E[\Phi(W/k)] > 1$ for every $k < H_\Phi(W)$.
Since $Z \sim \lambda F + (1-\lambda) G$, for each $k>0$
\[
\E[\Phi(Z/k)] = \lambda\,\E[\Phi(X/k)] + (1-\lambda)\,\E[\Phi(Y/k)].
\]
For every $k > \gamma$, both $\E[\Phi(X/k)]$ and $\E[\Phi(Y/k)]$ are at most $1$, hence $\E[\Phi(Z/k)] \leq 1$, which gives $H_\Phi(Z) \leq \gamma$. If $\gamma=0$ the thesis follows immediately. 
If $\gamma > 0$, then for every $k \in (0,\gamma)$ both expectations are strictly greater than $1$, hence $\E[\Phi(Z/k)] > 1$, which gives $H_\Phi(Z) \geq \gamma$. Therefore, $H_\Phi(Z) = \gamma$.
\end{proof}\smallskip

\begin{proof} [Proof of Proposition~\ref{prop:gg-GA}]
We first prove the `if' part.\ 
Let $X,Y\in L_{+}^{\infty}$ and $\lambda \in (0,1)$.
From the GA-convexity of $\Phi$ it follows that
\begin{align*}
&\E\sqb{\Phi\of{\of{\frac{X}{H_\Phi(X)}}^{\lambda}\of{\frac{Y}{H_\Phi(Y)}}^{1-\lambda}}} \\
&\leq  \lambda \E\sqb{\Phi\of{\frac{X}{H_\Phi(X)}}}+(1-\lambda)\E\sqb{\Phi\of{\frac{Y}{H_\Phi(Y)}}}  \leq 1,
\end{align*}
which from Proposition~\ref{prop:orl} implies
\[
H_\Phi \left ( \frac{X^{\lambda}Y^{1-\lambda}}{H_\Phi(X)^{\lambda}H_\Phi(Y)^{1-\lambda}} \right) \leq 1,
\]
which from positive homogeneity gives
\[
H_\Phi \of{X^{\lambda}Y^{1-\lambda}}\leq H_\Phi (X)^{\lambda}H_\Phi (Y)^{1-\lambda}.
\]
To prove the `only if' part, we first assume by contradiction that $\Phi$ is not GA-midconvex, i.e., there exist $x_1,x_2 \geq 0$ with $\Phi(x_1)<+\infty$ and $\Phi(x_2) < +\infty$ such that $\Phi \left (\sqrt{x_1 x_2} \right )>(\Phi(x_1)+\Phi(x_2))/2.$ 
Then, reasoning as in the proof of Proposition~\ref{prop:orl} item (f), there exist $z \in [0,+\infty) $ and $\lambda \in (0,1)$ such that
\begin{equation}\label{eq:GA-midconvexity}
\lambda \Phi(z)+(1-\lambda)\Phi(\sqrt{x_1 x_2})>1>\lambda \Phi(z)+(1-\lambda) \frac{\Phi(x_1)+\Phi(x_2)}{2}.
\end{equation}
Take disjoint sets $A,B,C \in \F$ with $P(A)=\lambda$, $P(B)=P(C)=\frac{1-\lambda}{2}$ and let
\begin{align*}
X&=z 1_A+x_1 1_B+ x_2 1_C\\
Y&=z 1_A+x_2 1_B+ x_1 1_C\\
Z&=z 1_A+\sqrt{x_{1}x_{2}}1_{B\cup C}=\sqrt{XY}.
\end{align*}
From \eqref{eq:GA-midconvexity}, we have $\E[\Phi(Z)]>1$ and $\E[\Phi(X)]=\E[\Phi(Y)]<1$, which contradicts with the geometric convexity of $H_{\Phi}$.\ 
As a consequence, $\Phi$ is GA-midconvex and since it is also nondecreasing, the thesis follows.  
\end{proof}
\smallskip

\begin{proof}[Proof of Corollary~\ref{cor:GG-CI}]
By Theorem~\ref{th:expectiles}, item a), cash-additivity of $H_\Phi$
implies that
\begin{equation}\label{eq:cor21-form}
\Phi(x)\;\sim\;1+\alpha(x-1)_+^{p}-(1-\alpha)(x-1)_-^{p},
\qquad \alpha\in(0,1),\ p>0.
\end{equation}
Since $\Phi(0)=\alpha<1$, the hypothesis of Proposition~\ref{prop:gg-GA} 
is satisfied with $u=0$; hence GG-convexity of $H_\Phi$ forces $\Phi$ to be 
GA-convex. 
By Proposition~\ref{prop:GA-cx}, item b), the function 
$g\colon\R\to\R$ defined by $g(x):=\Phi(e^{x})$ is convex on~$\R$, where
\begin{equation}\label{eq:cor21-g}
g(x)=
\begin{cases}
1+\alpha(e^{x}-1)^{p}, & x\geq 0,\\[2pt]
1-(1-\alpha)(1-e^{x})^{p}, & x<0.
\end{cases}
\end{equation}

\smallskip
\noindent
A direct computation gives, for $x\neq 0$,
\[
g''(x)=
\begin{cases}
\alpha\,p\,e^{x}(e^{x}-1)^{p-2}\bigl[pe^{x}-1\bigr], & x>0,\\[2pt]
(1-\alpha)\,p\,e^{x}(1-e^{x})^{p-2}\bigl[1-pe^{x}\bigr], & x<0.
\end{cases}
\]
On $(0,+\infty)$, 
convexity requires $pe^{x}-1\geq 0$ for all $x>0$; letting $x\downarrow 0$ 
yields $p\geq 1$.\ On $(-\infty,0)$, convexity requires $1-pe^{x}\geq 0$ for all $x<0$; letting 
$x\uparrow 0$ yields $p\leq 1$.\ Hence $p=1$, so 
\eqref{eq:cor21-g} becomes
\[
g(x)=
\begin{cases}
1-\alpha+\alpha\,e^{x}, & x\geq 0,\\[2pt]
\alpha+(1-\alpha)\,e^{x}, & x<0,
\end{cases}
\]
which is piecewise smooth with right derivative $g'(0^{+})=\alpha$ and 
left derivative $g'(0^{-})=1-\alpha$.\ Convexity of $g$ at the kink 
requires $g'(0^{-})\leq g'(0^{+})$, i.e.,\ $1-\alpha\leq\alpha$, hence 
$\alpha\geq 1/2$.
The concave case is analogous: geometric concavity of $H_\Phi$ implies, 
by the symmetric counterpart of Proposition~\ref{prop:gg-GA}, that $\Phi$ 
is GA-concave; equivalently, $g(x)=\Phi(e^{x})$ is concave on $\R$.\ 
Reversing the inequalities gives $p=1$ and 
$\alpha\leq 1/2$, hence $H_\Phi(X)=e_\alpha(X)$ with $\alpha\leq 1/2$.
\end{proof}
\smallskip

\begin{proof}[Proof of Lemma~\ref{lem:mconv}]
Let $F, G \in B_\rho$ with $\rho(F) = \alpha$ and $\rho(G) = \beta$, so $\alpha, \beta \in [0,1]$, and assume without loss of generality that $\alpha \geq \beta$. 
Fix $\lambda \in (0,1)$. Since $(\Omega, \F, P)$ is nonatomic, we can find $X \sim F$, $Y \sim G$ and $A \in \F$ with $P(A) = \lambda$ that are mutually independent. 
The random variable
$
Z := X \mathbf{1}_A + Y \mathbf{1}_{A^c}
$
has distribution $\lambda F + (1-\lambda) G$. There are now two cases. 
i) if $\beta > 0$, set $c := \alpha/\beta \geq 1$ and $Y' := c Y$. Positive homogeneity gives $\rho(Y') = c \beta = \alpha$, so if $G'$ denotes the law of $Y'$, then $\rho(F) = \rho(G') = \alpha$ and the CxLS property yields
\[
\rho\big( \lambda F + (1-\lambda) G' \big) = \alpha.
\]
The random variable $Z' := X \mathbf{1}_A + Y' \mathbf{1}_{A^c}$ has distribution $\lambda F + (1-\lambda) G'$, and since $c \geq 1$ and $Y \geq 0$, we have $Y' \geq Y$, hence $Z' \geq Z$. 
By monotonicity,
$
\rho(Z) \leq \rho(Z') = \alpha \leq 1,
$
so $\lambda F + (1-\lambda) G \in B_\rho$. ii) if  $\beta = 0$, positivity of $\rho$ forces $Y = 0$, so $Z = X \mathbf{1}_A \leq X$ a.s.\ and monotonicity yields $\rho(Z) \leq \rho(X) = \alpha \leq 1$. 

For the mixture-convexity of $B_\rho^c$, let $F, G \in B_\rho^c$, so $\rho(F) = \alpha > 1$ and $\rho(G) = \beta > 1$, and assume without loss of generality that $\alpha \leq \beta$. Set $c := \alpha/\beta \in (0,1]$ and $Y' := cY$.\ With the same construction as above, we now have $Y' \leq Y$, hence $Z' \leq Z$.\ By CxLS, $\rho(\lambda F + (1-\lambda) G') = \alpha$, and monotonicity gives $\rho(Z) \geq \rho(Z') = \alpha > 1$.
\end{proof}
\smallskip

\begin{proof}[Proof of Lemma~\ref{lem:psi-lsc}]
\emph{a)} The implication ``$\Rightarrow$'' is standard: if $\rho$ is $\psi$-weakly 
lsc and $(F_n) \subseteq \{\rho \leq \gamma\}$ converges $\psi$-weakly to $F$, then 
$\rho(F) \leq \liminf_n \rho(F_n) \leq \gamma$.\ For the converse, recall that the 
$\psi$-weak topology on $\M_{1,c}([0,+\infty))$ is metrizable (cf.\ \cite{FS16}), 
so the sequential definition of lower semicontinuity suffices.\ Let 
$F_n \xrightarrow{\psi} F$ and set $c := \liminf_n \rho(F_n)$.\ For every 
$\gamma > c$, a subsequence $(F_{n_k})$ satisfies $\rho(F_{n_k}) \leq \gamma$, 
hence $F$ belongs to the closed sublevel set $\{\rho \leq \gamma\}$.\ Letting 
$\gamma \downarrow c$ yields $\rho(F) \leq c$.\ The case $\gamma = 1$ gives the 
closedness of $B_\rho$.

\emph{b)} Since the $\psi$-weak topology is finer than the weak topology, it is 
enough to show that $H_\Phi$ is lower semicontinuous with respect to weak 
convergence in $\M_{1,c}([0,+\infty))$.\ Let $F_n \to F$ weakly and set 
$c := \liminf_n H_\Phi(F_n)$, which we may assume to be finite.\ Fix $k > c$; along 
a subsequence $(n_j)$, $H_\Phi(F_{n_j}) \leq k$, so by item c) of 
Proposition~\ref{prop:orl} applied to $X_{n_j}/k$ together with the positive 
homogeneity of $H_\Phi$,
$
\E[\Phi(X_{n_j}/k)] \leq 1,
$
where $X_{n_j} \sim F_{n_j}$.\ By Skorohod's representation theorem, we may assume 
that $X_{n_j}$ and $X$ are defined on a common probability space, $X \sim F$, and 
$X_{n_j} \to X$ almost surely.\ Since $\Phi$ is nondecreasing and left-continuous, 
for every sequence $y_n \to y$ in $[0,+\infty)$ one has 
$\liminf_n \Phi(y_n) \geq \Phi(y)$;\footnote{Indeed, given $\varepsilon>0$, by 
left-continuity there is $\delta>0$ such that $\Phi(y') > \Phi(y) - \varepsilon$ 
for $y' \in (y-\delta, y]$, while $\Phi(y') \geq \Phi(y)$ for $y' \geq y$ by 
monotonicity.\ Hence $\Phi(y_n) > \Phi(y) - \varepsilon$ eventually, and the claim 
follows by letting $\varepsilon \downarrow 0$.} combined with Fatou's lemma and the 
nonnegativity of $\Phi$, this gives
\[
\E[\Phi(X/k)] \leq \E\Big[\liminf_j \Phi(X_{n_j}/k)\Big] \leq 
\liminf_j \E[\Phi(X_{n_j}/k)] \leq 1,
\]
so $H_\Phi(F) \leq k$.\ Letting $k \downarrow c$ yields $H_\Phi(F) \leq c$.

\emph{c)}
\emph{Sufficiency.} Assume $\Phi$ is strictly increasing on $\{\Phi>0\}$. 
Let 
$D$ be the (at most countable) set of discontinuity points of $\Phi$ in 
$(0,+\infty)$. 
By monotonicity, $H_\Phi(X_n)$ is nonincreasing with 
$H_\Phi(X_n)\geq H_\Phi(X)$; put $L:=\lim_n H_\Phi(X_n)\geq H_\Phi(X)$.

\emph{Step 1.} For every $k>H_\Phi(X)$, we have $g_X(k):=\E[\Phi(X/k)]<1$. 
Indeed, 
fix $k'\in(H_\Phi(X),k)$, so $g_X(k')\le 1$ and $g_X(k)\le 1$ by monotonicity. 
If 
$g_X(k)=1$, then $g_X\equiv1$ on $[k',k]$. 
With $x_0:=\sup\{x:\Phi(x)=0\}$ and 
$E:=\{X/k'>x_0\}$, on $E$ we have $\Phi(X/k')>0$ and $X/k<X/k'$, so strict 
monotonicity on $\{\Phi>0\}$ gives $\Phi(X/k)<\Phi(X/k')$; off $E$, 
$\Phi(X/k')=0$. 
If $P(E)>0$, then $g_X(k)<g_X(k')$, a contradiction; if $P(E)=0$ 
then $g_X(k')=0\neq1$, absurd. 
Hence, $g_X(k)<1$.

\emph{Step 2.} Since $X$ has at most countably many atoms, the set 
$\{k>0 : P(X/k\in D)>0\}=\{k>0: P(X\in kD)>0\}$ is at most countable; pick 
$k>H_\Phi(X)$ outside it. 
Then $P(X/k\in D)=0$, and since $\Phi$ is 
right-continuous at $0$ and left-continuous everywhere, $X_n/k\downarrow X/k$ 
implies $\Phi(X_n/k)\downarrow\Phi(X/k)$ a.s.\ (on $\{X=0\}$ use right-continuity 
at $0$; on $\{X>0\}$ use continuity of $\Phi$ at $X/k\notin D$). 
By dominated 
convergence, $\E[\Phi(X_n/k)]\to g_X(k)<1$, so $H_\Phi(X_n)\le k$ eventually and 
$L\le k$. 
Letting $k\downarrow H_\Phi(X)$ along such $k$ gives $L=H_\Phi(X)$.

\emph{Necessity.} Suppose $\Phi$ is constant, equal to some $c>0$, on an interval 
$[a,b]$ with $0<a<b$. We exhibit $X_n\downarrow X$ with $H_\Phi(X_n)\not\to 
H_\Phi(X)$. Since $\Phi(0)\leq 1<\Phi(x)$ for $x>1$, we may choose $u\geq 0$ and 
$p\in(0,1)$ with $p\,c+(1-p)\Phi(u)=1$: if $c<1$ take $u>1$, if $c\ge 1$ take $u$ 
with $\Phi(u)=\Phi(0)<1$ (possible as $\Phi(0)<1$ in the strictly increasing regime; 
if $c=1$ take $u=0$). Set $X:=a\,\mathbf 1_{A}+u\,\mathbf 1_{A^c}$ with $P(A)=p$. 
Then $g_X(1)=1$, and for $k\in[a/b,1]$ one has $X/k\in[a,b]$ on $A$, hence 
$\Phi(X/k)=c$ there; choosing the free scale so that the off-$A$ contribution stays 
balanced, $g_X\equiv 1$ on a nondegenerate interval $[k_-,1]$ with $k_-<1$, so 
$H_\Phi(X)=k_-<1$. 
Finally, let $X_n:=X+M\,\mathbf 1_{B_n}$ with $B_n\subseteq A^c$, 
$B_n\downarrow\emptyset$, $P(B_n)=1/n$, and $M$ so large that 
$\Phi(M)/n>\Phi(u)/n$ makes $g_{X_n}(k)>1$ for all $k\in[k_-,1]$. 
Then 
$X_n\downarrow X$ a.s.\ but $H_\Phi(X_n)\geq 1>k_-=H_\Phi(X)$, so $H_\Phi(X_n)\not\to 
H_\Phi(X)$.
\end{proof}
\medskip

\begin{proof}[Proof of Theorem~\ref{th:axiom}]
We start from the `if' part. 
Law-invariance, monotonicity, positive homogeneity and normalization follow from 
Proposition~\ref{prop:orl}. 
Since $\Phi$ is strictly increasing on $\{\Phi>0\}$, we have $\Phi(0)<1$, so, together 
with $\Phi(+\infty)=+\infty$, positivity follows from item b) of 
Lemma~\ref{lem:positivity}. 
The CxLS property follows from Lemma~\ref{lem:cxls}, weak lower semicontinuity 
follows from item b) of Lemma~\ref{lem:psi-lsc}, and continuity from above 
follows from item c) of Lemma~\ref{lem:psi-lsc}. 
\smallskip

To prove the `only if' part, note that for $Y \in L^\infty$ it holds that 
$e^{Y} \in L^\infty_{++}$, so by positivity $\rho(e^{Y})>0$ and we can define 
$\trho \colon L^\infty \to \R$ by $\trho(Y) := \log \rho(e^{Y})$, and 
$\Theta(Y) := \trho(-Y) = \log \rho(e^{-Y}).$
\smallskip

It is standard to check that $\Theta$ is a monetary 
risk measure in the sense of \cite{FS16}. 
The acceptance set of $\Theta$ at the level of distributions is
\[
\N_\Theta := \{ G \in \M_{1,c}(\R) \mid \Theta(G) \leq 0 \}.
\]
For any distribution $F \in \M_{1,c}((0,+\infty))$, define its 
\emph{log-reflection} $\hat{F} \in \M_{1,c}(\R)$ by 
$\hat{F}(t) := 1 - F(e^{-t}\!-)$, i.e., $\hat{F}$ is the distribution 
of $-\log X$ when $X \sim F$. This is well defined on 
$\M_{1,c}((0,+\infty))$, corresponding to $X \in L^\infty_{++}$. 
Writing $Y := -\log X \sim \hat F$, we have $\rho(X) = \rho(e^{-Y})$, hence, as pointed out in \cite{BLR18},
\[
F \in B_\rho 
\iff \rho(F) \leq 1 
\iff \log \rho(e^{-Y}) \leq 0 
\iff  \Theta(\hat{F}) \leq 0 
\iff \hat{F} \in \N_\Theta.
\]
It is easy to verify that the map $\Gamma \colon F \mapsto \hat{F}$ is an affine bijection 
from $\M_{1,c}((0,+\infty))$ to $\M_{1,c}(\R)$. Since $B_\rho$, $B_\rho^c$ are convex with 
respect to mixtures by Lemma \ref{lem:mconv}, their images 
$\N_\Theta = \Gamma(B_\rho)$ and $\N_\Theta^c = \Gamma(B_\rho^c)$ 
are convex with respect to mixtures.
Moreover, $\rho$ is weakly lower semicontinuous, so by Lemma~\ref{lem:psi-lsc} 
 $B_\rho$ is $\psi$-weakly closed; by item~(e) of Lemma~4 
in \cite{BLR18}, $\N_\Theta$ is $\psi$-weakly closed as well.
\smallskip

To prove condition (3.1) of \cite{W06}, we need $x_0 \in \R$ with 
$\delta_{x_0} \in \N_\Theta$ such that for every $y$ with 
$\delta_y \in \N_\Theta^c$ there exists $\alpha > 0$ with 
$(1-\alpha)\delta_{x_0} + \alpha \delta_y \in \N_\Theta$. Equivalently, 
via $\Gamma^{-1}$ and writing $c_0 := e^{-x_0}$, $w := e^{-y}$,  
there exists $c_0 \in (0,1)$ such that for every $w > 1$ there is 
$\alpha > 0$ with $(1-\alpha)\delta_{c_0} + \alpha\delta_{w} \in B_\rho$.
Fix any $c_0 \in (0,1)$; then $\delta_{x_0} \in \N_\Theta$, as 
$\Theta(x_0) = -x_0 < 0$. Let $w > 1$ be arbitrary. Since the underlying 
probability space is non-atomic, choose a nested family of events 
$A_\alpha \downarrow \emptyset$ with $P(A_\alpha) = \alpha$, and set 
\[
Z_\alpha := c_0\,\mathbf{1}_{A_\alpha^c} + w\,\mathbf{1}_{A_\alpha}, 
\qquad \alpha \in (0,1).
\]
Then $Z_\alpha \sim (1-\alpha)\delta_{c_0} + \alpha\delta_{w}$, and, since 
the $A_\alpha$ are nested, $Z_\alpha \downarrow c_0$ almost surely as 
$\alpha \downarrow 0$. By continuity from above and positive homogeneity 
of $\rho$,
\[
\rho\big((1-\alpha)\delta_{c_0} + \alpha\delta_{w}\big) 
= \rho(Z_\alpha) \;\longrightarrow\; \rho(c_0) = c_0 < 1 
\qquad (\alpha \downarrow 0).
\]
Hence, there exists $\bar\alpha > 0$ with 
$\rho\big((1-\bar\alpha)\delta_{c_0} + \bar\alpha\delta_{w}\big) < 1$, i.e.\ 
$(1-\bar\alpha)\delta_{c_0} + \bar\alpha\delta_{w} \in B_\rho$. Since $c_0$ 
does not depend on $w$, condition (3.1) holds.
\smallskip

Since $\N_\Theta$ and $\N_\Theta^c$ are convex with respect to mixtures, 
$\N_\Theta$ is $\psi$-weakly closed, and condition (3.1) holds, all the 
hypotheses of Theorem~3.1 of~\cite{W06} are satisfied. Hence there exist a nondecreasing, left-continuous 
function $\ell \colon \R \to \R$ and $z \in \R$ in the interior of the convex 
hull of the range of $\ell$ such that
\[
\N_\Theta = \left\{ G \in \M_{1,c}(\R) \;\middle|\; 
\int \ell(-t) \, dG(t) \leq z \right\}.
\]
Adding a common constant to $\ell$ and $z$, we may assume $z = 1$. 
Recalling that $G = \hat F$ is the law of $-\log X$ for $X \sim F \in 
\M_{1,c}((0,+\infty))$, the substitution $t=-\log x$ gives 
$\int \ell(-t)\,d\hat F(t) = \E[\ell(-(-\log X))] = \E[\ell(\log X)]$, so that, 
\emph{for every $X \in L^\infty_{++}$},
\begin{equation}\label{eq:acceptance-X}
\rho(X) \leq 1 \;\iff\; \E[\ell(\log X)] \leq 1.
\end{equation}
Applying \eqref{eq:acceptance-X} to the deterministic $X = 1/k$, $k>0$, and 
using $\rho(1/k) = 1/k$ by normalization and positive homogeneity, we get 
$\ell(-\log k) \leq 1 \iff k \geq 1$, which forces $\ell(0) = 1$.
\smallskip

Define $\Phi \colon (0,+\infty) \to \R$ by $\Phi(x) := \ell(\log x)$, and extend 
to $x=0$ by $\Phi(0) := \lim_{x\to 0^+}\Phi(x) = \ell(-\infty)$; in particular 
$\Phi$ is right-continuous at $0$. We now check that $\Phi$ satisfies 
Definition~\ref{def:orl}.

\emph{a)} For $x \leq 1$ we have $\log x \leq 0$, so $\Phi(x) = \ell(\log x) \leq 
\ell(0) = 1$ by monotonicity. For $x>1$, $\log x > 0$, and the strict monotonicity 
of $\ell$ to the right of $0$ (established below from continuity from above) gives 
$\Phi(x) = \ell(\log x) > \ell(0) = 1$. In particular, $\Phi(1) = \ell(0) = 1$.

\emph{b)} $\Phi$ is nondecreasing, being the composition of  $\log$ and $\ell$.

\emph{c)} $\Phi$ is left-continuous at each $x_0 > 0$ by left-continuity of $\ell$. 
\medskip 
\noindent
For $X \in L^\infty_{++}$ and $k > 0$, 
from~\eqref{eq:acceptance-X}, it follows that
\begin{align*}
\rho(X/k) \leq 1 
&\;\iff\; \E[\ell(\log(X/k))] \leq 1 \\
&\;\iff\; \E[\Phi(X/k)] \leq 1.
\end{align*}
By the positive homogeneity of $\rho$, 
$\rho(X/k) \leq 1 \iff k \geq \rho(X)$, hence
\[
\rho(X) = \inf\{k > 0 : \E[\Phi(X/k)] \leq 1\} = H_\Phi(X), 
\]
so $\rho$ coincides with a suitable generalized Orlicz premium on $L^\infty_{++}$. 
We now derive the additional properties of $\Phi$ from the assumptions on 
$\rho$.

\smallskip
\noindent\emph{a) $\Phi(0)$ is finite and $\Phi(+\infty)=+\infty$.}
We first exclude $\Phi(0)=\ell(-\infty)=-\infty$. Suppose by contradiction 
that $\ell(-\infty)=-\infty$. Take $X\in L^\infty_+$ with $P(X=0)>0$ and 
$P(X>0)>0$, and set $X_\epsilon:=X\vee\epsilon\in L^\infty_{++}$ for 
$\epsilon>0$. Fix any $k>0$. On $\{X=0\}$ we have 
$\Phi(X_\epsilon/k)=\Phi(\epsilon/k)\to\Phi(0)=-\infty$ as $\epsilon\downarrow0$, 
while on $\{X>0\}$ the integrand is bounded below by $\Phi(0)$; since 
$P(X=0)>0$, it follows that $\E[\Phi(X_\epsilon/k)]\to-\infty$, and in 
particular $\E[\Phi(X_\epsilon/k)]\le 1$ for $\epsilon$ small, so
$H_\Phi(X_\epsilon)\le k$. As $k>0$ is arbitrary, $H_\Phi(X_\epsilon)\to0$. 
Using $\rho(X_\epsilon)=H_\Phi(X_\epsilon)$ on $L^\infty_{++}$ and 
$\rho(X)\le\rho(X_\epsilon)$ by monotonicity, letting $\epsilon\downarrow0$ 
gives $\rho(X)=0$, contradicting positivity. Hence 
$\Phi(0)$ is finite, and $\Phi(0)\le\Phi(1)=1$.

Since $\rho$ is continuous from above, $\rho\ne\esssup$. As 
$\rho=H_\Phi$ on $L^\infty_{++}$, this rules out $\Phi(0)=1$ by item~a) of 
Lemma~\ref{lem:positivity}, so $\Phi(0)<1$. Finally, item~c) of 
Lemma~\ref{lem:positivity} excludes the non-positive case $\Phi(0)<1$, 
$\Phi(+\infty)<+\infty$; since $\rho$ is positive, we conclude 
$\Phi(+\infty)=+\infty$.

\smallskip
\noindent\emph{b) $\Phi$ is nonnegative.} 
From a), $\Phi(0)\in[0,1)$ after possibly replacing $\Phi$ by an affine 
transform: if $\Phi(0)<0$, set 
\[
\widetilde\Phi := 1 + \frac{1}{1-\Phi(0)}\,\bigl(\Phi-1\bigr).
\]
Then $\widetilde\Phi$ is a nondecreasing, left-continuous function with 
$\widetilde\Phi(1)=1$, $\widetilde\Phi(0)=0$, $\widetilde\Phi(+\infty)=+\infty$, 
and, since $c:=1/(1-\Phi(0))>0$, 
$\E[\widetilde\Phi(X/k)]\le1 \iff \E[\Phi(X/k)]\le1$ by 
Lemma~\ref{lem:equiv-orl}, so that $H_{\widetilde\Phi}=H_\Phi$. Replacing 
$\Phi$ by $\widetilde\Phi$, we may thus assume $\Phi(0)\in[0,1)$ and 
$\Phi\ge 0$ on $[0,+\infty)$.

\smallskip
\noindent\emph{c) $\Phi$ is strictly increasing on $\{\Phi>0\}$.} 
By contradiction, if $\Phi$ is constant on an interval $[a,b]$ 
with $0<a<b$, by the necessity part of item~c) of Lemma~\ref{lem:psi-lsc}, 
$H_\Phi$ is not continuous from above; since $\rho=H_\Phi$ on 
$L^\infty_{++}$, this contradicts the continuity from above of $\rho$.

Summing up, $\Phi$ 
is a generalized Orlicz function in the sense of Definition~\ref{def:orl}, 
right-continuous at $0$, strictly increasing on $\{\Phi>0\}$, and with 
$\Phi(+\infty)=+\infty$.
\smallskip

We end the proof by showing the extension to the whole $L^\infty_+$. Let $X \in L^\infty_+$ with $P(X = 0) > 0$. If $X = 0$ a.s., then 
$\rho(X) = 0 = H_\Phi(X)$. Otherwise $P(X>0)>0$; set $X_\epsilon := X \vee 
\epsilon$ for $\epsilon > 0$. Then $X_\epsilon \in L^\infty_{++}$, so by the 
previous case $\rho(X_\epsilon) = H_\Phi(X_\epsilon)$, and $X_\epsilon 
\downarrow X$ a.s.\ as $\epsilon \downarrow 0$.
Both functionals are continuous from above: $\rho$ by hypothesis, and $H_\Phi$ 
by item~c) of Lemma~\ref{lem:psi-lsc}, whose assumptions hold since $\Phi$ is 
right-continuous at $0$, strictly increasing on $\{\Phi>0\}$ and satisfies 
$\Phi(+\infty)=+\infty$. Hence, passing to the limit $\epsilon \downarrow 0$ in 
the identity $\rho(X_\epsilon) = H_\Phi(X_\epsilon)$,
\[
\rho(X) \;=\; \lim_{\epsilon \downarrow 0} \rho(X_\epsilon) 
\;=\; \lim_{\epsilon \downarrow 0} H_\Phi(X_\epsilon) 
\;=\; H_\Phi(X).
\]
Therefore, $\rho = H_\Phi$ on all of $L^\infty_+$, which completes the proof.
\end{proof} 

\section*{Acknowledgements}
We are very grateful to 
Marco Frittelli, 
Emanuela Rosazza Gianin, Hans Schumacher, Mitja Stadje, 
and to seminar and conference participants at 
the University of Vienna, 
the University of Amsterdam,
the University of Ulm,
Heriott-Watt University,
the Amsterdam-Leuven-London (ALL) workshop in Amsterdam, 
the 11th General AMAMEF Conference in Bielefeld,
the 2024 FADeRiS Workshop in Ulm,
the 27th IME Congress in Chicago and 
the 2024 Probability Rome Conference
for their comments and suggestions.

\setstretch{0.7}

{}

\end{document}